\documentclass{article}
\usepackage[utf8]{inputenc}
\usepackage{natbib}
\usepackage[colorlinks]{hyperref}
\usepackage{amsmath,amssymb,amsthm}
\usepackage[margin=2cm]{geometry}
\usepackage{tikz,tikz-cd}
\usepackage{float}
\usepackage{enumerate}
\usepackage{parskip}
\usepackage[color=yellow]{todonotes}
\usepackage[noblocks]{authblk}

\newtheorem{theorem}{Theorem}

\newtheorem{definition}{Definition}[section]

\newtheorem{lemma}[theorem]{Lemma}

\newtheorem{remark}[theorem]{Remark}

\def\contract{\makebox[1.2em][c]{\mbox{\rule{.6em}
{.01truein}\rule{.01truein}{.6em}}}}

\begin{document}

\title{Stochastic Geometric Mechanics for Fluid Dynamics}
\author[1a,1b]{Darryl D. Holm} 
\author[2a,2b]{Erwin Luesink}

\affil[1a]{Department of Mathematics, Imperial College London, SW7 2BX London, United Kingdom}
\affil[1b]{d.holm@imperial.ac.uk, ORCID ID: 0000-0001-6362-9912}
\affil[2a]{Department of Mathematics, Faculty EEMCS, University of Twente, PO Box 217, 7500 AE Enschede, The Netherlands}
\affil[2b]{e.luesink@utwente.nl, ORCID ID: 0000-0003-1804-6800}
\date{}

\maketitle
\begin{abstract}
Stochastic geometric mechanics (SGM) is known for its potential utility in quantifying uncertainty in  global climate modelling of the Earth's ocean and atmosphere while also preserving the fundamental advective transport properties of ideal fluid flow. The present chapter describes the mathematical development of the framework of stochastic geometric mechanics in the context of fluid flow and wave dynamics obtained from Hamilton's variational principle.
\end{abstract}

\section*{Key words}
\begin{center}
$\bullet\,$ Geometric mechanics $\qquad$ $\bullet\,$ Fluid dynamics $\qquad$ $\bullet\,$ Stochastic partial differential equations\\
$\bullet\,$ Lie groups $\qquad$ $\bullet\,$ Diffeomorphism group $\qquad$ $\bullet\,$ Sobolev spaces $\qquad$ $\bullet\,$ Momentum maps
\end{center}

\tableofcontents

\section{Introduction}
%Some say that a fluid is a liquid, gas or other material that deforms %continuously when an external force is applied, others say that a fluid is a %substance without fixed shape that yields easily to external pressure. While %these are certainly useful definitions, they do not reveal what happens in the %absence of external forces. 

Euler’s classic theory of ideal fluid dynamics represents the motion of all the fluid particles in a container as a time-dependent curve on the manifold of volume-preserving smooth invertible maps now called diffeomorphisms. Moreover, this curve describing the sequential actions of the volume-preserving diffeomorphisms on the fluid domain is a special optimal curve that distills the fluid motion into a single statement. Namely, “A fluid moves to get out of its own way as efficiently as possible.” Put more mathematically, Euler fluid flow occurs along a time-dependent curve in the manifold of volume-preserving diffeomorphisms which is a geodesic with respect to the metric on its tangent space supplied by its kinetic energy, as noted by \cite{arnold1966geometrie}.

%Geodesics determine the shortest path between two points given a certain metric. This was the first time %that fluid motion could be considered within the framework of geometric mechanics. 

Thus, the problem of determining the kinematic motion of an incompressible fluid was transformed by \cite{arnold1966geometrie} into a geometric mechanics problem which can be expressed mathematically as a variational principle. Variational principles are not restricted to the kinematic motion of ideal fluids. If the forces arise from a potential, then variational principles can also be employed to derive the equations of compressible fluids with advected quantities. The Lagrangian is the key to extending fluid models beyond the kinematic case. The Lagrangian is a functional that describes energetic interactions obtained by transforming between kinetic and potential energy. If the Lagrangian is invariant under the action of a Lie group, then Noether's theorem \cite{noether1918invariante} implies the existence of a conservation law corresponding to each Lie symmetry. This means that the state space of the fluid can be split into two sets: variables that express dynamics; and functions of the dynamical variables that express Noether conservation laws. Fluids possess the particle relabelling symmetry, for example, which permits the description of fluids by means of both particles (Lagrangian point of view) and fields (Eulerian point of view). The particle relabelling symmetry of the Eulerian representation induces an infinite family of conservation laws defined on closed loops of fluid material moving with the flow. These conservation laws are derived from the Kelvin-Noether theorem for fluid circulation in geometric mechanics, as shown, e.g., in \cite{holm1998euler,cotter2013noether}. 
%Hence a fluid could also be defined as a dynamical system that satisfies the Kelvin circulation theorem.

Geometric mechanics is particularly useful in the context of geophysical fluid dynamics (GFD). This is because GFD applies at planetary scales such as Earth's ocean and atmosphere. Motion at planetary scales involves large masses whose fluid motions are essentially unaffected by viscosity. Hence, provided that the GFD models are energetically closed, Hamilton's principle is able to produce GFD models whose underlying geometry helps to understand conserved quantities and can be used as a guide for numerical discretisation. However, formulating appropriate GFD models for atmospheric or oceanic processes remains a challenging problem, since one must deal with vast ranges of spatial and temporal scales whose associated dynamical interactions involve a large number of disparate and often unknown processes. Such a lack of information can be interpreted as a representation error. Representation errors are not the only source of modelling errors, though. Together with numerical errors and observation errors, many physical parameterisations such as the air-sea interaction dynamics contain uncontrolled approximations. The many sources of uncertainty arising from incomplete information 
motivate the usage of stochastic parametrisations. However, these stochastic parametrisations need to preserve the fundamental geometric structure that underlies the GFD models, particularly the Kelvin-Noether circulation theorem. Stochastic parametrisations that preserve this geometric structure can be achieved via stochastic geometric mechanics.

\section{Noether's theorem and the momentum map}
\subsection*{Key points}
\begin{center}
    $\bullet\,$ Lie symmetries $\qquad$
    $\bullet\,$ Noether's theorem $\qquad$
    $\bullet\,$ Momentum map $\qquad$
    $\bullet\,$ Reduction by symmetry
\end{center}
Geometric mechanics deals with group-invariant variational principles. Its origin goes back to the early 1900s, where \cite{poincare1901forme} showed that when a Lie algebra acts locally transitively on the configuration space $M$ of a Lagrangian mechanical system, then the Euler-Lagrange equations are equivalent to a new system of differential equations on the product of the configuration space with the Lie algebra. These equations are now known as the Euler-Poincar\'e equations and play a crucial role in fluid dynamics. The original work of Poincar\'e is presented in modern language in the paper \cite{marle2013henri}. 

Euler-Poincar\'e equations arise when the action of a Lie group $G$ with Lie algebra $\mathfrak{g}$ on a manifold $M$ is locally transitive. In terms of the tangent lift action $G\times TM\to TM$ of a Lie group $G$ on the tangent bundle $TM$ of a manifold (or configuration space) $M$ on which $G$ acts transitively, Noether's theorem states that each continuous symmetry of a Lagrangian $L:TM\to \mathbb{R}$ defined in the action integral $S=\int L(q,v)dt$ for Hamilton's variational principle $\delta S = 0$ with $(q,v)\in TM$ implies a conserved quantity for the corresponding Euler-Lagrange equations.

The conserved quantities arising from Noether's theorem, \cite{noether1918invariante}, are dual to the infinitesimal symmetries. These dual quantities are momenta and the map $J:T^*M\to\mathfrak{g}^*$ that transforms variables from the cotangent bundle $T^*M$ to the dual $\mathfrak{g}^*$ of the Lie algebra $\mathfrak{g}$ associated with the Lie group $G$ is known as the \emph{momentum map}. The momentum map is the central object in geometric mechanics. It was introduced in various different ways and levels of generality by \cite{kirillov1962unitary}, \cite{kostant1970quantization}, \cite{souriau1970structure} and \cite{smale1970topologya, smale1970topologyb}.

The dynamical variable $m\in\mathfrak{g}^*$ in the dual Lie algebra is the momentum. In general, the configuration manifold $M$ is not a Lie group. However, when a Lie group $G$ acts transitively on a configuration manifold $M$ the proof of Noether's theorem induces a cotangent-lift momentum map $J: T^*M\to\mathfrak{g}^*$. The momentum map induced this way is an infinitesimally equivariant Poisson map taking functions on the cotangent bundle $T^*M$ of $M$ to the dual Lie algebra $\mathfrak{g}^*$ of the Lie group $G$. One can also consider momentum maps that are not the cotangent lift momentum map, when the momentum map $J:P\to \mathfrak{g}^*$, where $P$ is any symplectic manifold. In this more general setting, the momentum map is not necessarily ${\rm Ad}^*$-equivariant and cocycles start to play a role. This is the setting of symplectic mechanics, as introduced by \cite{souriau1970structure}. The lack of equivariance plays a role when one considers the group of diffeomorphisms over the circle, as we will show in section \ref{sec:circle}.

The cotangent lift momentum map $J: T^*M\to\mathfrak{g}^*$ is equivariant and Poisson, even if $G$ is not a Lie symmetry of the Lagrangian in Hamilton's principle. Momentum maps naturally lead from the Lagrangian side to the Hamiltonian side. The Hamiltonian dynamics on $T^*M$ involves symplectic transformations. However, as we shall discuss below, for the class of Hamiltonians which can be defined as $H\circ J: \mathfrak{g}^*\to \mathbb{R}$, the momentum map induces Euler-Poincar\'e motion on the Lagrangian side and Lie-Poisson motion on the Hamiltonian side. The momentum map connects the Hamiltonian reduction techniques of  \cite{marsden1974reduction} with the Lagrangian reduction techniques of \cite{holm1998euler}. To illustrate this equivalence, we return to the situation in which the configuration manifold, $M$, is a Lie group, $G$.

One may reconstruct the solution on $G$ from its representation on $T^*G\setminus G\simeq\mathfrak{g}^*$. In that case, solving the equations describing the evolution of the momentum map on the dual Lie algebra $\mathfrak{g}^*$ is equivalent to solving the equations on the cotangent bundle $T^*G$ when the configuration manifold is $G$. When the Lie group $G$ acts transitively, freely and properly on the configuration manifold $M$, then one may reconstruct the solution on $M$ from its representation on $T^*G\setminus G\simeq\mathfrak{g}^*$. The last statement is proved for finite-dimensional Lie groups $G$ in, e.g., \cite{abraham1978foundations}. Provided the Lagrangian or Hamiltonian is hyperregular, the Legendre transform is a diffeomorphism. The Euler-Poincar\'e reduction procedure can then be expressed in terms of the cube of linked commutative diagrams shown in figure \ref{fig:cube}. 
\begin{figure}[H]
\small
\centering
\begin{tikzcd}[row sep=3em, column sep=small]
& 
L:TG\to\mathbb{R} \arrow[dl] \arrow[rr,  "\text{Legendre transform}", leftrightarrow] \arrow[dd] 
& 
& 
H:T^*G\to\mathbb{R} \arrow[dl] \arrow[dd]
\\
\text{Euler-Lagrange eqns} \arrow[rr, crossing over, Leftrightarrow] 
& 
& \text{Hamilton's eqns}
\\
& \ell:\mathfrak{g}\to\mathbb{R} \arrow[dl] \arrow[rr, "\text{Legendre \hspace{0.25cm} transform}", leftrightarrow] 
& 
& \hslash:\mathfrak{g}^*\to\mathbb{R} \arrow[dl]
\\
\text{Euler-Poincar\'e eqns} \arrow[rr, Leftrightarrow] \arrow[from=uu, crossing over]
& 
& \text{Lie-Poisson eqns} \arrow[from=uu, crossing over]
\end{tikzcd}
\caption{The cube of commutative diagrams for geometric mechanics on Lie groups. Euler-Poincar\'e reduction (on the left side) and Lie-Poisson reduction (on the right side) are both indicated by the arrows pointing down. The diagrams are all commutative, provided the Legendre transformation and reduced Legendre transformation are both invertible.}
\label{fig:cube}
\end{figure}
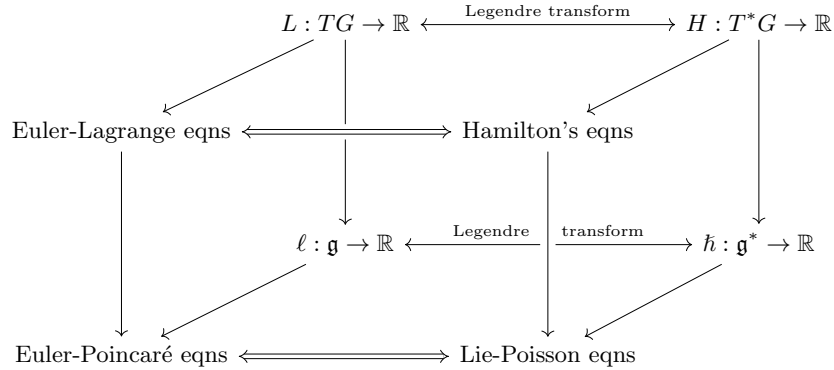
The notation in Figure \ref{fig:cube} is as follows: $G$ denotes the configuration manifold which is assumed to be isomorphic to a Lie group, $TG$ is the tangent bundle, $T^*G$ is the cotangent bundle, $TG\setminus G \simeq \mathfrak{g}$ is the Lie algebra and $T^*G\setminus G\simeq \mathfrak{g}^*$ is the dual of the Lie algebra. The Lagrangian is a functional $L:TG\to\mathbb{R}$ and the Hamiltonian is a functional $H:T^*G\to\mathbb{R}$. Euler-Poincar\'e reduction takes advantage of Lie group symmetries to transform the Lagrangian and Hamiltonian into group-invariant variables, which leads to a reduced Lagrangian $\ell:\mathfrak{g}\to\mathbb{R}$ and a reduced Hamiltonian $\hslash:\mathfrak{g}^*\to\mathbb{R}$. The diagram comprising the face of the cube involving these functionals in figure \ref{fig:cube} commutes if the Legendre transform is a diffeomorphism. This is guaranteed if the Lagrangian or Hamiltonian is hyperregular. The Euler-Lagrange equations and Hamilton's equations are related via an invertible change of variables, which also holds for the Euler-Poincar\'e equations and the Lie-Poisson equations. The invertibility of change of variables follows from the Legendre transform being a diffeomorphism. Many finite dimensional mechanical systems may be described naturally in this framework. The classic example is the rotating rigid body, discussed from the viewpoint of symmetry reduction by Poincar\'e in \cite{poincare1901forme}. In his 1901 paper, Poincar\'e also raised the issue of \emph{symmetry breaking}, by introducing the vertical acceleration of gravity, which breaks the  $SO(3)$  symmetry for free rotation and restricts it to  $SO(2)$ for rotations about the vertical axis. 

Figure \ref{fig:cube} describes geometric mechanics in the finite dimensional setting without symmetry breaking, but cannot describe the infinite dimensional setting due to presence of the mass density, which breaks symmetry. We will show this in the next section, where we investigate the extension of Figure \ref{fig:cube} to the setting of broken symmetries and stochasticity. We describe the geometric mechanics with the diffeomorphism group first for generic Riemannian manifolds, then we introduce stochasticity and we specialise to diffeomorphisms over the circle in section \ref{sec:circle}. 

Stochasticity can be included in the framework of Euler-Poincar\'e reduction by symmetry. The first attempt to include noise consistently in finite-dimensional symplectic Hamiltonian mechanics was by \cite{bismut1982mecanique} and reduction by symmetry of stochastic systems was studied by \cite{lazaro2008stochastic}. In the present work, we will review Euler-Poincar\'e reduction of stochastic infinite dimensional variational systems with symmetry breaking. The infinite dimensional case is interesting because it is the natural setting for fluid dynamics, quantum mechanics and elasticity. We will explore the infinite dimensional case in context of geophysical fluid dynamics, where symmetry under the smooth invertible maps of the flow domain is broken by the spatial dependence of the initial mass density. For general smooth manifolds, the geometry of the diffeomorphism group over those manifolds is poorly understood with a few exceptions. One such exception is the group of diffeomorphisms over the circle, which we will investigate in the stochastic setting.

\section{Sobolev class diffeomorphisms}
\subsection*{Key points}
\begin{center}
    $\bullet\,$ Diffeomorphism group $\qquad$
    $\bullet\,$ Riemannian structure $\qquad$
    $\bullet\,$ Differential forms $\qquad$
    $\bullet\,$ Duality and vector fields
\end{center}
The discovery of \cite{arnold1966geometrie} initiated intensive study into groups of diffeomorphisms over compact manifolds. An extensive and detailed review can be found in \cite{smolentsev2007diffeomorphism}. Here we will recall some of the key results. Consider an $n$-dimensional compact and oriented smooth manifold $M$ without boundary, equipped with a smooth Riemannian metric $g$. Common examples of manifolds without boundary are spheres and tori, which appear naturally in geophysical fluid dynamics and plasma physics. The Riemannian structure induces a pointwise inner product on the tangent bundle $TM$, i.e., for $X,Y\in T_xM$ we denote the inner product $g(x)(X,Y)$ by $\langle X,Y\rangle_x$. This can be extended in the usual way to an inner product on the bundles $T^p_q M$, where we denote the inner product by the same symbol. If $U,V\in \Gamma(T_q^pM)$ are two tensor fields, then their inner product is given by 
\begin{equation}\label{eq:l2inner}
    (U,V) = \int_M \langle U(x),V(x)\rangle_x\, d\mu(x).
\end{equation}
This inner product induces an $L^2$-type topology on the space of tensor fields. It is weak because the $C^\infty$-topology of uniform convergence of all derivatives is a natural topology on the space of smooth tensor fields and is stronger than the $L^2$-type. One can also equip the space of tensor fields with the $C^k$-topology, which is defined for a tensor field $U$ as
\begin{equation}
    |U|_k = \sum_{i=0}^k \sup_{x\in M}\sqrt{\langle\nabla^{(i)}U,\nabla^{(i)}U\rangle_x},
\end{equation}
where $\nabla^{(i)} = \nabla\circ\cdots\circ \nabla$ is the $i$th power of the covariant derivative and the norm $\|\,\cdot\,\|_x$ is the square root of the pointwise inner product. The $C^\infty$-topology is defined by the family of norms $|U|_k$ for $k\geq 0$, and in this case the space of tensor fields is a Fr\'echet space. For Fr\'echet spaces, there does not exist an inverse or implicit function theorem, nor does there exist a general solution theory for ordinary differential equations. The space of tensor fields of class $C^k$ is a Banach space with respect to the norm $|\,\cdot\,|_k$, which is much better than the Fr\'echet setting. However, for many problems it is more convenient to have a Hilbert space structure. This is achieved by considering an inner product that is stronger than the $L^2$-inner product \eqref{eq:l2inner}. Let $s\geq 0$ and let $U,V\in \Gamma(T^p_q M)$, then the Sobolev class tensor fields are given by
\begin{equation}\label{eq:hsinner}
    (U,V)_s = \sum_{i=0}^s \int_M \left\langle \nabla^{(i)}U,\nabla^{(i)}V\right\rangle_x \,d\mu(x).
\end{equation}
Let $\Gamma^s(T_q^p M)$ be the completion of $\Gamma(T^p_q M)$ with respect to the topology induced by \eqref{eq:hsinner}. By the Sobolev embedding theorems one has the usual relations between spaces of varying degrees of smoothness. The reason that one can still study the differential geometry of spaces equipped with Sobolev class topologies without having to count derivatives is due to the the following theorem, which can be found in \cite{palais1966foundations}.
\begin{theorem}
    Let $E$ and $F$ be two vector bundles over $M$ and let $f:E\to F$ be a fiber-preserving smooth mapping. If $s\geq n/2+1$, then the mapping $\phi:\Gamma^s(E)\to\Gamma^s(F)$ defined by $\phi(\alpha) = f\circ\alpha$ is a map of class $C^\infty$.
\end{theorem}
To study fluid dynamics from a topological point of view, one considers the spaces $C^\infty(M,N)$, $C^k(M,N)$ and $H^s(M,N)$ as infinite dimensional manifolds. The atlas of coordinate charts for these spaces in constructed in a canonical way, using the Riemannian structure of $M$ and $N$. This construction gives rise to atlas called the \emph{natural atlas}. To speak of diffeomorphism groups, these spaces need to be equipped with a product. Here one needs to be careful, since groups of diffeomorphisms are not Lie groups, which can be shown by means of the $\alpha-\omega$-theorems of \cite{abraham1963lectures}.
\begin{theorem}[$\alpha$]
    Let $h:M\to M'$ be a map of class $C^k$, then the mapping
    \begin{equation}
        \alpha_h:C^k(M',N)\to C^k(M,N), \qquad \alpha_h(f) = f\circ h,
    \end{equation}
    is a map of class $C^\infty$.
\end{theorem}

\begin{theorem}[$\omega$]
    Let $h:N\to N'$ be a map of class $C^{k+l}$, then the mapping
    \begin{equation}
        \omega_h:C^k(M,N)\to C^k(M,N'), \qquad \omega_h(f) = h\circ f,
    \end{equation}
    is a map of class $C^l$.
\end{theorem}
For $s>n/2+1$, the space of Sobolev class diffeomorphisms is defined as the set
\begin{equation}
    \mathfrak{D}^s(M) = H^s(M,M)\cap \{ \eta\in C^1(M,M)\,|\, \eta \text{ is bijective and } \eta^{-1}\in C^1(M,M)\}.
\end{equation}
When there is no danger of confusion, we will write $\mathfrak{D}^s$ instead of $\mathfrak{D}^s(M)$.  Analogous to the $\alpha-\omega$-theorems for $C^k$-mappings, \cite{ebin1967space} showed that the right action 
\begin{equation}
    R_\eta:\mathfrak{D}^s\to\mathfrak{D}^s, \qquad f\mapsto f\circ \eta,
\end{equation}
is a mapping of class $C^\infty$ for any $\eta\in\mathfrak{D}^s$, whereas the left action, for $\eta\in \mathfrak{D}^{s+l}$
\begin{equation}
    L_\eta:\mathfrak{D}^s\to\mathfrak{D}^s, \qquad f\mapsto \eta\circ f,
\end{equation}
is merely a mapping of class $C^l$. Thus, if one restricts to right shifts, it is not necessary to count derivatives and one can perform differential topology and geometry as usual. So we consider the manifold $M$ to be acted upon by a group of Sobolev class diffeomorphisms $\mathfrak{D}^s = \{\eta\in H^s(M,M)|\,\eta\text{ is bijective and } \eta^{-1}\in H^s(M,M)\}$. The space of Sobolev class diffeomorphisms is both a Hilbert manifold and a topological group provided that $s>n/2+1$. The Hilbert manifold structure is much better for analysis than the Fr\'echet manifold structure, since for Hilbert manifolds one has the existence of function inverses and the implicit function theorem, as well as the existence of a general solution theorem for ordinary differential equations. This additional structure also implies that one can construct the tangent space of $\mathfrak{D}^s$ in the usual way and study geodesics. At this point, one can introduce dynamics. 

The space $\mathfrak{D}^s$ is the configuration space for continuum mechanics and each $\eta\in\mathfrak{D}^s$ is called a configuration. A fluid trajectory starting from $x_0\in M$ at time $t=0$ is given by $x(t)=\eta_t(x_0)=\eta(x_0,t)$, with $\mathfrak{D}^s\ni \eta:M\times\mathbb{R}^+\to M$ a continuous one-parameter subgroup of $\mathfrak{D}^s$. In the deterministic case, computing the time derivative of this one-parameter subgroup yields the \emph{reconstruction equation}, given by
\begin{equation}
\frac{\partial}{\partial t}\eta_t(x_0) = u(\eta_t(x_0),t),
\label{eq:reconstructiondeterministic}
\end{equation}
where $u_t(\,\cdot\,)=u(\,\cdot\,,t)\in \mathfrak{X}^s(M)$ is a time dependent vector field with flow $\eta_t(\,\cdot\,)=\eta(\,\cdot\,,t)$. The initial data is given by $\eta(x_0,0)=x_0$. One can view the space of vector fields $\mathfrak{X}^s(M)$ as the space Sobolev class sections of the tangent bundle, i.e., $\mathfrak{X}^s(M)=\Gamma^s(TM)$. However, this hides the important fact that the vector fields are a special case of a more general construction, that of bundle-valued differential forms. The dual of a vector field is much more intuitive in the framework of bundle-valued differential forms. Bundle-valued differential forms are constructed in the following manner. Let $M$ be a smooth manifold and let $E\to M$ be a smooth vector bundle over $M$. Recall that a scalar-valued differential $k$-form is given by $\Omega^k(M) = \Gamma(\bigwedge^k T^*M)$, where the sections are considered smooth. A Sobolev-class $E$-valued differential $k$-form on $M$ is a multilinear map of class $H^s$ that associates to each $x\in M$ a element of $\bigwedge^k T_x^*M \otimes E_x$. This is a differential $k$-form on $M$ with values in $E$. The space of Sobolev-class $E$-valued differential $k$-forms is given by 
\begin{equation}
\Omega^{k,s}(M,E) = \Gamma^s(E\otimes\textstyle\bigwedge^k T^*M).
\end{equation}
Let $E^*$ denote the dual of $E$, induced in the canonical way, then the dual of the space of Sobolev-class $E$-valued differential $k$-forms is the space of Sobolev-class $E^*$-valued differential $n-k$-forms
\begin{equation}
\Omega^{k,s}(M,E)^* = \Gamma^s(E^*\otimes\textstyle\bigwedge^{n-k} T^*M)= \Omega^{n-k,s}(M,E^*).
\end{equation}
The duality pairing between a Sobolev-class $E$-valued $k$-form $\zeta$ and an $E^*$-valued $n-k$-form $\mathcal{X}$ is given by
\begin{equation}
\langle \mathcal{X},\zeta\rangle = \int_M\zeta\,\dot{\wedge}\,\mathcal{X},
\end{equation}
where the wedge-dot is the operator
\begin{equation}
\dot{\wedge}:\Omega^{k,s}(M,E)\times\Omega^{l,s}(M,E^*)\to \Omega^{k+l}(M).
\end{equation}
Some words of caution are necessary here. While the pullback distributes over the usual wedge product, it does not for the wedge-dot product. The pullback does distribute over the form part of the wedge-dot product. Another important remark is that the De Rham complex of bundle-valued differential forms is only exact if the underlying manifold is flat. The space of Sobolev-class vector fields is then $\mathfrak{X}^s(M) = \Omega^{0,s}(M;TM)=\Gamma^s(TM)$, the space of vector-valued scalar fields. This means that the dual of the vector fields is naturally identified with covector-valued densities, i.e. $\mathfrak{X}^s(M)^* = \Omega^{n,s}(M;T^*M)$, as we will now show.

The space $\mathfrak{D}^s$ inherits a \emph{weak Riemannian structure} from the underlying manifold $M$ in a natural way. For $\eta\in \mathfrak{D}^s$ and $U_\eta,V_\eta\in T_\eta\mathfrak{D}^s$, one can define the following bilinear form
\begin{equation}
(U_\eta,V_\eta)_\eta =\int_M \langle U_\eta(x),V_\eta(x)\rangle_{\eta(x)}\,d\mu(x),
\label{eq:weakriemannian}
\end{equation}
where $\mu$ is the volume form on $M$ induced by the metric. Note the subtle difference with the bilinear form in \eqref{eq:l2inner}. This bilinear form is neither left- nor right-invariant. To combat this issue, one defines $U=dR_\eta^{-1} U_\eta$ and $V=dR_\eta^{-1} V_\eta$. Then $U,V\in T_e\mathfrak{D}^s$ and pulling back the volume form by $\eta^{-1}$ gives
\begin{equation}\label{eq:weakriemannian2}
(U_\eta,V_\eta)_\eta = \int_M \langle U(x),V(x)\rangle_x (\eta^{-1})^*(d\mu).
\end{equation}
The Riemannian structures \eqref{eq:weakriemannian},\eqref{eq:weakriemannian2} are weak because the topology is of type $L^2$, which is strictly weaker than the $H^s$ topology. The Riemannian structure \eqref{eq:weakriemannian2} is a linear functional on the Hilbert space $T_\eta\mathfrak{D}^s$ and canonically induces the dual space $T_\eta^*\mathfrak{D}^s$. The pairing between $U\in T_\eta\mathfrak{D}^s$ and $\alpha\in T_\eta^*\mathfrak{D}^s$ is given by 
\begin{equation}
\langle \alpha,U \rangle = \int_M \alpha(x)\cdot U(x).
\label{eq:dualitypairing}
\end{equation}
Hence the metric on $M$ and the volume form $d\mu(x)$ can be used to construct the isomorphism between $T\mathfrak{D}^s$ and $T^*\mathfrak{D}^s$ as $U(x)\mapsto\alpha(x)=U^\flat(x)d\mu(x)$, where $\flat:TM\to T^*M$ is one of the musical isomorphisms that are induced by the metric on $M$. This shows that the dual $\alpha$ of a vector field $U$ is a covector-valued density, with $U^\flat$ being the covector part and $d\mu$ the density part. The group $\mathfrak{D}^s$ is not a Lie group; since right multiplication is smooth, but left multiplication is only continuous. Hence $\mathfrak{D}^s$ is a topological group with a weak Riemannian structure. In general, these properties are not sufficient to guarantee the existence of an exponential map. However, \cite{ebin1970groups} showed that an exponential map exists in many important cases. In particular, they showed that the geodesic spray associated to \eqref{eq:weakriemannian} (with and without forcing) is smooth.\footnote{The \emph{geodesic spray} is the vector field whose integral curves are the geodesics.} The smoothness of the geodesic spray persists even though $H^s$ diffeomorphisms are considered rather than smooth diffeomorphisms. Combined with the existence of an exponential map, the smoothness property implies a regular interpretation of the Euler-Poincar\'e equations on $\mathfrak{D}^s$, provided that one uses right translations and right representations of the group on itself and its Lie algebra, as shown in \cite{holm1998euler}. However, due to the presence of the mass form $\rho\, d\mu = (\eta^{-1})^*d\mu = \eta_*(d\mu)$, the bilinear form \eqref{eq:weakriemannian2}, which will serve to define the kinetic energy, is not right-invariant under the action of the entire $H^s$ diffeomorphism group, although there is right-invariance under the action of the isotropy subgroup $\mathfrak{D}^s_\mu = \{\eta\in\mathfrak{D}^s | \, \eta_*(d\mu) = d\mu\}$. This subgroup is a proper subgroup, because it is smaller than $\mathfrak{D}^s$ itself. Thus, one speaks of \emph{symmetry breaking} due to so-called advected quantities.

\begin{definition}[Advected quantity]
A fluid variable is said to be \emph{advected}, if it keeps its value along Lagrangian particle trajectories. Advected quantities are sometimes called \emph{tracers}, because the evolution histories of scalar advected quantities with different initial values (labels) trace out the Lagrangian particle trajectories of each label, or initial value, via the \emph{push-forward} of the full diffeomorphism group, i.e., $a_t=\eta_{t\,*}a_0= (\eta_t^{-1})^*(a_0)$, where $\eta_t$ is a time-dependent curve on the manifold of diffeomorphisms that represents the fluid flow.
\end{definition}\smallskip

\begin{remark}[Advected quantities as order parameters]
When several advected quantities are involved, the space $V^*$ is the direct sum of several vector spaces, where each summand space hosts a different advected quantity. In general, each additional advected quantity decreases the dimension of the isotropy subgroup. For example, consider an ideal deterministic fluid with a buoyancy variable $b$, then the Lagrangian corresponding to the model will depend on the mass form $\rho\,d\mu = (\eta^{-1})^*(d\mu)$ and $b$ in a parametric manner. This Lagrangian will be right-invariant under the action of the isotropy subgroup $\mathfrak{D}^s_{\mu,b} = \{\eta\in\mathfrak{D}^s|\, \eta_*(d\mu)=d\mu \text{ and } \eta_*b=b\}$. Hence, advected quantities are \emph{order parameters} and each additional order parameter breaks more symmetry. For the sake of notation, one usually writes $\mathfrak{D}^s_{a_0}$ for the isotropy subgroup, no matter how many advected quantities there are. One then uses $a$ to represent all advected quantities and $a_0$ to denote the initial value of the advected quantities. 
\end{remark}

\subsection{Semidirect product group adjoint \& coadjoint actions}\label{sec:adjcoadj}
\subsection*{Key points}
\begin{center}
    $\bullet\,$ Semidirect product groups $\qquad$ 
    $\bullet\,$ Lie algebra and its dual $\qquad$
    $\bullet\,$ Adjoint and coadjoint actions \\
    $\bullet\,$ Lie-Poisson bracket $\qquad$
    $\bullet\,$ Diamond operator
\end{center}
The semidirect product group action is constructed in the following way. The representation of $\mathfrak{D}^s$ on a vector space $V$ is by push-forward, which is a left representation, as shown by \cite{marsden1984semidirect}. The representation of the group on itself and on its Lie algebra is a right representation. In terms of analysis, this means that all representations are smooth and no derivatives need to be counted. The group action of the semidirect product group is given by
\begin{equation}
\begin{aligned}
\bullet:(\mathfrak{D}^s\times V)\times(\mathfrak{D}^s\times V)\to(\mathfrak{D}^s\times V)\\
\quad (\eta_1,v_1)\bullet(\eta_2,v_2):= (\eta_1\circ \eta_2,v_2+(\eta_2)_*v_1)
\end{aligned}
\label{eq:semidirectproduct}
\end{equation}
with $\eta_1,\eta_2\in\mathfrak{D}^s$ and $v_1,v_2\in V$. Since this is a right-action, it is natural to read this from right to left. The semidirect product group is often denoted as $\mathfrak{D}^s\circledS V = (\mathfrak{D}^s\times V,\bullet)$. In the group action above, $(\eta_2)_*v_1$ denotes the \emph{push-forward} of $v_1$ by $\eta_2$ and $\circ$ denotes composition. Note that the group affects both slots in \eqref{eq:semidirectproduct}, but the vector space only appears in the second slot. The identity element of the semidirect product group is $(e,0)$ where $e\in\mathfrak{D}^s$ is the identity diffeomorphism and $0\in V$ is the zero vector. An inverse element is given by
\begin{equation}
(\eta,v)^{-1} = (\eta^{-1},-(\eta^{-1})_*v) = (\eta^{-1}, -\eta^*v),
\end{equation}
where $\eta^*v$ denotes the pull-back of $v$ by $\eta$. To understand how reduction works for semidirect products, it is helpful to know how the group acts on its Lie algebra and on the dual of its Lie algebra. Duality will be defined with respect to the sum of the pairing \eqref{eq:dualitypairing} and the dual linear transformation on $V$. This results in the pairing 
\begin{equation}
\langle(m,a),(u,b)\rangle_{(\mathfrak{X}^s\times V)^*\times(\mathfrak{X}^s\times V)} = \langle m,u \rangle_{\mathfrak{X}^{s*}\times\mathfrak{X}^s} + \langle a,b \rangle_{V^*\times V}
\end{equation}
of $\mathfrak{X}^s\times V$ and its dual. Consider two at least $C^1$ one-parameter subgroups $(\eta_t,v_t),(\widetilde{\eta}_\epsilon,\widetilde{v}_\epsilon)\in \mathfrak{D}^s\times V$. Using these one-parameter subgroups, one can compute the inner automorphism, or adjoint action of the group on itself. This adjoint action is defined by conjugation
\begin{equation}
\begin{aligned}
{\rm AD}:(\mathfrak{D}^s\times V)\times(\mathfrak{D}^s &\times V)\to(\mathfrak{D}^s\times V),\\
{\rm AD}_{(\eta_t,v_t)}(\widetilde{\eta}_\epsilon,\widetilde{v}_\epsilon) &:= (\eta_t,v_t)\bullet (\widetilde{\eta}_\epsilon,\widetilde{v}_\epsilon)\bullet (\eta_t,v_t)^{-1}\\
&= \big(\eta_t\circ\widetilde{\eta}_\epsilon\circ \eta_t^{-1}, \eta_t^*(\widetilde{v}_\epsilon - v_t + \widetilde{\eta}_{\epsilon*}v_t)\big).
\end{aligned}
\label{eq:AD}
\end{equation}
To see how the group acts on its Lie algebra, one can compute the derivative with respect to $\epsilon$ and evaluate at $\epsilon=0$ in the adjoint action of the group on itself. Let $\mathfrak{X}^s\ni \widetilde{u}=\frac{d}{d\epsilon}|_{\epsilon=0}\widetilde{\eta}_\epsilon$ and $V\ni\widetilde{b}=\frac{d}{d\epsilon}|_{\epsilon=0}\widetilde{v}_\epsilon$. This choice for a vector field is guided by the deterministic reconstruction equation in \eqref{eq:reconstructiondeterministic}. For any tensor $S_\epsilon\in T_s^r(M)$ whose dependence on $\epsilon$ is at least $C^1$ it holds that
\begin{equation}
\frac{d}{d\epsilon}\widetilde{\eta}_{\epsilon*}S_\epsilon = \widetilde{\eta}_{\epsilon*}\left(\frac{d}{d\epsilon}S_\epsilon-\mathcal{L}_{\widetilde{u}} S_\epsilon\right).
\label{eq:liechainrule}
\end{equation}
Important here is that the Lie derivative does not commute with pull-backs and push-forwards that depend on parameters, see \cite{abraham1978foundations}. The adjoint action of the group on its Lie algebra can be computed as 
\begin{equation}
\begin{aligned}
{\rm Ad}:(\mathfrak{D}^s\times V)\times(\mathfrak{X}^s &\times V)\to (\mathfrak{X}^s\times V),\\
{\rm Ad}_{(\eta_t,v_t)}(\widetilde{u},\widetilde{b})&:= \frac{d}{d\epsilon}\Big|_{\epsilon=0}{\rm AD}_{(\eta_t,v_t)}(\widetilde{\eta}_\epsilon,\widetilde{v}_\epsilon)\\
&= (\eta_{t*}\widetilde{u},\eta^*_t\widetilde{b}-\eta_t^*\mathcal{L}_{\widetilde{u}}v_t).
\end{aligned}
\label{eq:Ad}
\end{equation}
By means of the pairing on $\mathfrak{X}^s\times V$, one can compute the dual action to the adjoint action \eqref{eq:Ad}. This is called the coadjoint action of the group on the dual of its Lie algebra and it is a representation only when one defines using the inverse of a group element. Let $(\widetilde{m},\widetilde{a})\in(\mathfrak{X}^s\times V)^*$, then the coadjoint action is given by
\begin{equation}
\begin{aligned}
{\rm Ad}^*:(\mathfrak{D}^s\times V)\times(\mathfrak{X}^s &\times V)^*\to(\mathfrak{X}^s\times V)^*,\\
\langle{\rm Ad}^*_{(\eta_t^{-1},-\eta_t^{-1}v_t)}(\widetilde{m},\widetilde{a}),(\widetilde{u},\widetilde{b})\rangle &:= \langle(\widetilde{m},\widetilde{a}),{\rm Ad}_{(\eta_t,v_t)}(\widetilde{u},\widetilde{b})\rangle,\\
{\rm Ad}^*_{(\eta_t^{-1},-\eta_t^{-1}v_t)}(\widetilde{m},\widetilde{a}) &= (\eta_t^*\widetilde{m}+v_t\diamond \eta_{t*}\widetilde{a},\eta_{t*}\widetilde{a}).
\end{aligned}
\label{eq:Ad*}
\end{equation}
\begin{definition}[The diamond operator]
The coadjoint action \eqref{eq:Ad*} features the \emph{diamond operator}, which is defined for $a\in V^*$, $u\in\mathfrak{X}^s$ and fixed $v\in V$ as
\begin{equation}
\langle v\diamond a, u\rangle_{\mathfrak{X}^{s*}\times\mathfrak{X}^s} := -\langle a,\mathcal{L}_u v\rangle_{V^*\times V}.
\end{equation}
Note that the diamond operator is the dual of the Lie derivative regarded as a map $\mathcal{L}_{(\,\cdot\,)}v:\mathfrak{X}^s\to V$, hence $v\diamond(\,\cdot\,):V^*\to\mathfrak{X}^{s*}$. The diamond operator shows how an element from the dual of the vector space acts on the dual of the Lie algebra. 
\end{definition}

When evaluated at $t=0$, the $t$-derivatives of ${\rm Ad}$ in \eqref{eq:Ad} and ${\rm Ad}^*$ in \eqref{eq:Ad*} define, respectively, the adjoint and coadjoint actions of the Lie algebra on itself and on its dual. Denote by $\mathfrak{X}^s\ni u = \frac{d}{dt}|_{t=0}\eta_t$ and $V\ni b=\frac{d}{dt}|_{t=0}v_t$. The adjoint action of the Lie algebra on itself is
\begin{equation}
\begin{aligned}
{\rm ad}:(\mathfrak{X}^s\times V)\times(\mathfrak{X}^s &\times V)\to (\mathfrak{X}^s\times V),\\
{\rm ad}_{(u,b)}(\widetilde{u},\widetilde{b})&:=\frac{d}{dt}\Big|_{t=0}{\rm Ad}_{(g_t,v_t)}(\widetilde{u},\widetilde{b}),\\
{\rm ad}_{(u,b)}(\widetilde{u},\widetilde{b})&=(-\mathcal{L}_u\widetilde{u},\mathcal{L}_u\widetilde{b}-\mathcal{L}_{\widetilde{u}}b)\\
&= (-[u,\widetilde{u}],\mathcal{L}_u\widetilde{b}-\mathcal{L}_{\widetilde{u}}b),
\end{aligned}
\label{eq:ad}
\end{equation}
where the bracket $[\,\cdot\,,\,\cdot\,]$ in \eqref{eq:ad} is the commutator of vector fields. The minus sign is due to fact that group acts on itself from the right. The coadjoint action of the Lie algebra on its dual can be obtained by computing the dual to \eqref{eq:ad} or by taking the derivative with respect to $t$ and evaluate at $t=0$ in \eqref{eq:Ad*}. Either way, one arrives at
\begin{equation}
\begin{aligned}
{\rm ad}^*:(\mathfrak{X}^s\times V)\times(\mathfrak{X}^s &\times V)^*\to(\mathfrak{X}^s\times V)^*,\\
\langle{\rm ad}^*_{(u,b)}(\widetilde{m},\widetilde{a}),(\widetilde{u},\widetilde{b})\rangle &:= \langle(\widetilde{m},\widetilde{a}),{\rm ad}_{(u,b)}(\widetilde{u},\widetilde{b})\rangle,\\
{\rm ad}^*_{(u,b)}(\widetilde{m},\widetilde{a}) &= (\mathcal{L}_u\widetilde{m} + b\diamond\widetilde{a},-\mathcal{L}_u\widetilde{a}),
\end{aligned}
\label{eq:ad*}
\end{equation}
in which \eqref{eq:ad} implies the last line in \eqref{eq:ad*}. Alternatively, one can also obtain \eqref{eq:ad*} by taking the derivative with respect to $t$ in \eqref{eq:Ad*} and evaluate at $t=0$.

In summary, the adjoint and coadjoint actions and operators for semidirect product groups of diffeomorphisms and spaces of advected quantities are given by
\begin{table}[H]
\centering
\begin{tabular}{l|ll}
${\rm AD}:(\mathfrak{D}^s\times V)\times(\mathfrak{D}^s\times V)\to(\mathfrak{D}^s\times V)$ & ${\rm AD}_{(\eta,v)}:$ & $(\widetilde{\eta},\widetilde{v}) \mapsto (\eta\circ\widetilde{\eta}\circ\eta^{-1},\, \eta^*(\widetilde{v}-v+\widetilde{\eta}_*v)$ \\
${\rm Ad}:(\mathfrak{D}^s\times V)\times (\mathfrak{X}^s\times V)\to(\mathfrak{X}^s\times V)$ & ${\rm Ad}_{(\eta,v)}:$ & $(\widetilde{u},\widetilde{b})\mapsto (\eta_*\widetilde{u},\, \eta^*\widetilde{b}-\eta^*\mathcal{L}_{\widetilde{u}}v)$ \\
${\rm Ad}^*:(\mathfrak{D}^s\times V)\times (\mathfrak{X}^s\times V)^*\to(\mathfrak{X}^s\times V)^*$ & ${\rm Ad}^*_{(\eta,v)^{-1}}:$ & $(\widetilde{m},\widetilde{a})\mapsto (\eta^*\widetilde{m} + v\diamond\eta_*\widetilde{a},\, \eta_*\widetilde{a})$ \\
${\rm ad}:(\mathfrak{X}^s\times V)\times(\mathfrak{X}^s\times V)\to(\mathfrak{X}^s\times V)$ & ${\rm ad}_{(u,b)}:$ & $(\widetilde{u},\widetilde{b})\mapsto (-[u,\widetilde{u}],\, \mathcal{L}_u\widetilde{b}-\mathcal{L}_{\widetilde{u}}b)$\\
${\rm ad}^*:(\mathfrak{X}^s\times V)\times(\mathfrak{X}^s\times V)^*\to(\mathfrak{X}^s\times V)^*$ & ${\rm ad}^*_{(u,b)}:$ & $(\widetilde{m},\widetilde{a}) \mapsto (\mathcal{L}_u\widetilde{m} + b\diamond \widetilde{a},\,-\mathcal{L}_u\widetilde{a})$
\end{tabular}
\caption{The adjoint and coadjoint actions of the semidirect product group made of the diffeomorphism group and a vector bundle and the corresponding semidirect product Lie algebra.}
\label{tab:semidirect}
\end{table}
\smallskip

\begin{remark}[Coadjoint action and the diamond operator]
The coadjoint action is an important operator in geometric mechanics and representation theory. It was shown by \cite{kirillov1962unitary} and in further work by \cite{kostant1970quantization} and \cite{souriau1970structure} that the coadjoint orbits of a Lie group $G$ have the structure of symplectic manifolds and are connected with Hamiltonian mechanics. See \cite{kirillov1999merits} for a review. The computations of the adjoint and coadjoint actions for the semidirect product group is valuable for fluid mechanics, as they introduce the two fundamental operators that appear in the equations of motion. The Lie derivative is responsible for transport of tensors along vector fields and its dual action given by the diamond operator encodes the symmetry breaking. In particular, the diamond operator introduces the effect of symmetry breaking into the Euler-Poincar\'e equations of motion.
\end{remark}

\subsection{Geometric mechanics with diffeomorphisms}
\subsection*{Key points}
\begin{center}
    $\bullet\,$ Lagrangians and Hamiltonians $\qquad$
    $\bullet\,$ Legendre transform $\qquad$ \\
    $\bullet\,$ Reduction by symmetry $\qquad$
    $\bullet\,$ Reconstruction equation 
\end{center}
With the adjoint and coadjoint actions defined, one can derive continuum mechanics equations with advected quantities by using symmetry reduction. Euler-Poincar\'e reduction for a semidirect product group $\mathfrak{D}^s\times V$ as developed in \cite{holm1998euler} is sketched below in figure \ref{fig:cubesdp}. 
\begin{figure}[H]
\small
\centering
\begin{tikzcd}[row sep=3em, column sep=small]
& 
L:T\mathfrak{D}^s\times V^*\to\mathbb{R} \arrow[dl] \arrow[rr,  "\text{Legendre transform}", leftrightarrow] \arrow[dd] 
& 
& 
H:T^*(\mathfrak{D}^s\times V)\to\mathbb{R} \arrow[dl] \arrow[dd]
\\
\text{Euler-Lagrange eqns} \arrow[rr, crossing over, Leftrightarrow] 
& 
& \text{Hamilton's eqns}
\\
&
\ell:\mathfrak{X}^s\times V^*\to\mathbb{R} \arrow[dl] \arrow[rr, "\text{Legendre \hspace{0.25cm} transform}", leftrightarrow] 
& 
& 
\hslash:(\mathfrak{X}^s\times V)^*\to\mathbb{R} \arrow[dl] 
\\
\text{Euler-Poincar\'e eqns} \arrow[rr, Leftrightarrow] \arrow[from=uu, crossing over]
& 
& 
\text{Lie-Poisson eqns} \arrow[from=uu, crossing over]
\end{tikzcd}
\caption{The cube of continuum mechanics in the semidirect product group setting. Reduction is indicated by the arrows pointing down.}
\label{fig:cubesdp}
\end{figure}
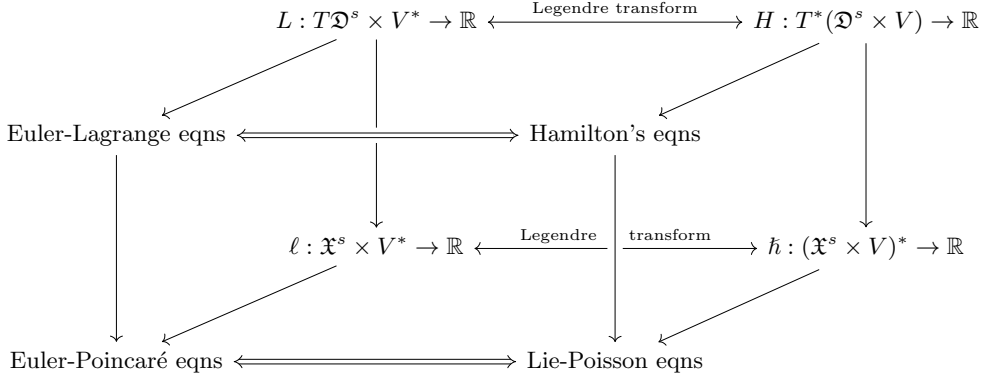
 As shown by comparison of  figure \ref{fig:cubesdp} with figure \ref{fig:cube}, several new features arise in semidirect product Lie group reduction which differ from Euler-Poincar\'e reduction by symmetry when the configuration space itself is a Lie group. These differences can be conveniently explained by introducing the physical concept of an order parameter. As discussed earlier, the order parameters in continuum mechanics are the elements of $V^*$ which are advected by the action of the diffeomorphism group $\mathfrak{D}^s$. The advection is defined simply as the semidirect product action on the elements of $V^*$. The introduction of each additional advected state variable (or, order parameter) into the physical problem shrinks the symmetry group $\mathfrak{D}^s$. The remaining symmetry of the Lagrangian in Hamilton's principle is the isotropy subgroup $\mathfrak{D}^s_{a_0}$ of the initial conditions, $a_0$, for the entire set of advected quantities, $a$. The action of the diffeomorphism group $\mathfrak{D}^s$ on these initial conditions then describes their advection as the action of $\mathfrak{D}^s$ on its coset space $\mathfrak{D}^s\setminus\mathfrak{D}^s_{a_0}=V^*$. Once the inital values of the order parameters, $a_0$,  have been set, one must still define a Legendre transform from the Lagrangian formulation into the Hamiltonian formulation and vice versa. The Legendre transform in the setting of semidirect products is a partial Legendre transform, since it transforms between $T\mathfrak{D}^s$ and $T^*\mathfrak{D}^s$ or $T\mathfrak{D}^s\setminus\mathfrak{D}^s_{a_0} \simeq \mathfrak{X}^s$ and $T^*\mathfrak{D}^s\setminus\mathfrak{D}^s_{a_0} \simeq\mathfrak{X}^{s*}$ only after having fixed the value $a_0$ of the order parameters, which live in $V^*$. This coset reduction is what figure \ref{fig:cubesdp} shows. The remaining right-invariance of a functional under the action of the isotropy subgroup is called its \emph{particle relabelling symmetry}.

Our exploration continues on the Lagrangian side in figure \ref{fig:cubesdp}. Consider a Lagrangian $L:T\mathfrak{D}^s\times V^*\to\mathbb{R}$. By fixing the value of $a_0\in V^*$, one can construct $L_{a_0}:T\mathfrak{D}^s\to\mathbb{R}$. If this Lagrangian is right-invariant under the action of the isotropy subgroup $\mathfrak{D}_{a_0}^s$, then one can construct 
\begin{equation}
\begin{aligned}
L\left(\frac{d}{dt}\eta\circ \eta^{-1},e,a_0\right) &= L_{a_0}\left(\frac{d}{dt}\eta\circ \eta^{-1},e\right)\\
&= \ell_{a_0}\left(\frac{d}{dt}\eta\circ \eta^{-1}\right) = \ell\left(\frac{d}{dt}\eta\circ \eta^{-1}, \eta_*a_0\right).
\end{aligned}
\label{eq:lagrangians}
\end{equation}
Here $\circ$ means composition of functions. The same procedure applies to the Hamiltonian. Since the coadjoint action is known, it is straightforward to formulate the Lie-Poisson equations. The details of Hamiltonian semidirect product reduction and also more information on the Lagrangian semidirect product reduction can be found in \cite{holm1998euler}. 

The coadjoint action of the Lie algebra on its dual is also required for the Lagrangian semidirect product reduction. One can use the deterministic reconstruction equation to see that the argument of the Lagrangians in \eqref{eq:lagrangians} is
\begin{equation}
\frac{d}{dt}\eta\circ \eta^{-1} = u.
\end{equation}
Using this information, one can integrate the Lagrangian in time to construct the action functional. By requiring the variational derivative of the action functional to vanish, one can compute the equations of motion. However, due to the removal of symmetries, the variations are no longer free. In the next section, we replace the deterministic reconstruction equation by a semimartingale and formulate the stochastic Euler-Poincar\'e theorem.

\section{Stochastic geometric mechanics with diffeomorphisms}
\subsection*{Key points}
\begin{center}
    $\bullet\,$ Semimartingale $\qquad$
    $\bullet\,$ Stratonovich integral $\qquad$
    $\bullet\,$ Stochastic Lie-chain rule
\end{center}
Geometric mechanics can be made stochastic in two distinct ways while preserving structure. The first option is called stochastic advection by Lie transport (SALT) and was introduced by \cite{holm2015variational}. SALT leaves the geometric structure invariant, but at the cost of losing conservation of energy. It corresponds to replacing the deterministic reconstruction equation in \eqref{eq:reconstructiondeterministic} by the semimartingale
\begin{equation}
{\sf d}\eta_t(x_0) = u(\eta_t(x_0),t)dt + \sum_{i=1}^M \xi_i(\eta_t(x_0))\circ dW_t^i,
\label{eq:reconstructionstochastic}
\end{equation}
where here the symbol $\circ$ means that the stochastic integral is taken in the Stratonovich sense. Note that in all other instances $\circ$ means composition of functions. The initial data is given by $\eta(x_0,0)=x_0$. The $W_t^i$ are independent, identically distributed Brownian motions, defined with respect to the standard stochastic basis $(\Omega,\mathcal{F},(\mathcal{F}_t)_{t\geq 0},\mathbb{P})$, see \cite{karatzas1998brownian}. Such a noise was shown to arise from a multi-time homogenisation argument in \cite{cotter2017stochastic}. This replacement \eqref{eq:reconstructionstochastic} need not be by a semimartingale, one can introduce geometric rough paths by the same approach, which was done in \cite{crisan2022variational}. The $\xi_i(\,\cdot\,)\in\mathfrak{X}^s$ are called data vector fields and are prescribed. These data vector fields represent the effects of unresolved degrees of freedom on the resolved scales of motion and account for unrepresented processes. 

The data vector fields $\xi_i$ can be determined by applying empirical orthogonal function analysis to appropriate numerical and/or observational data. For instance, for an application to the two dimensional Euler equations for an ideal fluid, see \cite{cotter2019numerically} and \cite{ephrati2023data}. An application of this framework to a two-layer quasi-geostrophic model can be found in \cite{cotter2018modelling}. Stochastic models enable the use of a variety of methods in data assimilation, which are discussed in \cite{cotter2019particle}. It is not difficult to make sense of \eqref{eq:reconstructiondeterministic}, but understanding \eqref{eq:reconstructionstochastic} is more complicated and requires stochastic analysis. Of particular importance is a stochastic chain rule, which is shown to exist in \cite{de2020implications}. This stochastic chain rule is called the \emph{Kunita-It\^o-Wentzell (KIW) formula} and helps interpret the semimartingale in \eqref{eq:reconstructionstochastic}. The KIW formula will also be used later to prove the stochastic Kelvin circulation theorem. 

In deriving equations in continuum mechanics, one needs to keep track of the mass form as well. As discussed in the previous section, an appropriate mathematical setting for this is an \emph{outer semidirect product group}. This means that one constructs a new group from two given groups with a particular type of group operation. For continuum mechanics, the ingredients are $\mathfrak{D}^s$ and $V^*$, where $V^*$ is a vector space of tensor fields. The reason for starting with $V^*$ rather than just $V$ is historical and will become clear when we discuss the extension of the diagram in Figure \ref{fig:cube}. The vector space $V^*$ is the space of \emph{advected quantities} and it will always contain at least the mass form $\rho\,d\mu=(\eta^{-1})^*(d\mu)$.  \smallskip

\subsection{Stochastic Euler-Poincar\'e theorem}
\subsection*{Key points}
\begin{center}
    $\bullet\,$ Variational derivative $\qquad$
    $\bullet\,$ Symmetry-reduced variations $\qquad$
    $\bullet\,$ Stochastic Euler-Poincar\'e theorem $\qquad$ \\
    $\bullet\,$ Stochastic Kelvin-Noether theorem $\qquad$
    $\bullet\,$ Stochastic Euler-Boussinesq equations
\end{center}
In the situation where noise is present, that is, when the reconstruction equation is  \eqref{eq:reconstructionstochastic}, the Euler-Poincar\'e variations become stochastic. Consider $\eta:\mathbb{R}^2\to\mathfrak{D}^s$ with $\eta_{t,\epsilon}=\eta(t,\epsilon)$ to be a two parameter subgroup with smooth dependence on $\epsilon$, but only continuous dependence on $t$. Let us denote 
\[
{\sf d}\chi_{t,\epsilon}(X) = ({\sf d}\eta_{t,\epsilon}\circ \eta_{t,\epsilon})(X) = u_{t,\epsilon}(X)dt + \sum_{i=1}^N \xi_i(X)\circ dW_t^i
\]
 and  
 \[
 v_{t,\epsilon}(X) = (\frac{\partial}{\partial \epsilon}\eta_{t,\epsilon}\circ \eta_{t,\epsilon})(X)\,.
 \] 
When a $\circ$ symbol is followed by $dW_t$ it means Stratonovich integration and in every other context the $\circ$ symbol is used to denote composition. Note that the data vector fields $\xi_i$ are prescribed and hence will not have a dependence on $\epsilon$. 
 
In order to compute with these stochastically parametrised subgroups and their associated vector fields, one needs a stochastic Lie chain rule. The Kunita-It\^o-Wentzell (KIW) formula is the stochastic generalisation of the Lie chain rule \eqref{eq:liechainrule}. A proof of the KIW formula is given  in \cite{de2020implications} for differential $k$-forms and vector fields. The proof includes the technical details on regularity that will be omitted here. In the KIW formula, the $k$-form is allowed to be a semimartingale itself. Let $K$ be a continuous adapted semimartingale that takes values in the $k$-forms and satisfies
\begin{equation}
K_t = K_0 + \int_0^t G_s ds + \sum_{i=1}^N\int_0^t H_{i\,s}\circ dB_s^i,
\label{eq:kformsemimartingale}
\end{equation}
where the $B_t^i$ are independent, identically distributed Brownian motions. The drift of the semimartingale $K$ is determined by $G$ and the diffusion by $H_i$, both of which are $k$-form valued continuous adapted semimartingales with suitable regularity. Let $\eta_t$ satisfy \eqref{eq:reconstructionstochastic}, then \cite{de2020implications} shows that the following holds
\begin{equation}
{\sf d}(\eta_t^*K_t) = \eta_t^*\big({\sf d}K_t + \mathcal{L}_{u_t} K_t\,dt + \mathcal{L}_{\xi_i}K_t \circ dW_t^i\big).
\label{eq:kiwformula}
\end{equation}
Equation \eqref{eq:kformsemimartingale} helps to interpret the ${\sf d}K_t$ term in the KIW formula \eqref{eq:kiwformula}. This formula will be particularly useful in computing the variations of the variables in the Lagrangian. To compute these variations, one needs the variational derivative. 

\paragraph{The variational derivative.} The variational derivative of a functional $F:\mathcal{B}\to\mathbb{R}$, where $\mathcal{B}$ is a Banach space, is denoted $\delta F/\delta \rho$ with $\rho\in\mathcal{B}$. The variational derivative can be defined by the first variation of the functional
\begin{equation}
\delta F[\rho]:= \frac{d}{d\epsilon}\Big|_{\epsilon=0} F[\rho+\epsilon \delta\rho] = \int \frac{\delta F}{\delta \rho}(x)\delta\rho(x)\,dx = \left\langle\frac{\delta F}{\delta \rho},\delta \rho\right\rangle.
\end{equation}
In the definition above, $\epsilon\in\mathbb{R}$ is a parameter, $\delta\rho\in\mathcal{B}$ is an arbitrary function and the first variation can be understood as a Fr\'echet derivative. A precise and rigorous definition can be found in \cite{gelfand2000calculus}. With the definition of the functional derivative in place, the following lemma can be formulated.
\medskip

\begin{lemma}
With the notation as above, the variations of $u$ and any advected quantity $a$ are given by 
\begin{equation}
\delta u(t) = {\sf d}v(t) + [{\sf d}\chi_t,v(t)],\quad \delta a(t) = -\mathcal{L}_{v(t)}a(t),
\label{def:delta-var}
\end{equation}
where $v(t)\in\mathfrak{X}^s$ is arbitrary.
\end{lemma}
\begin{proof}
The proof of the variation of $a(t)$ is a direct application of the Kunita-It\^o-Wentzell formula to $a(t,\epsilon)=g_{t,\epsilon*}a_0$. Note that the data vector fields $\xi_i$ are prescribed and do not depend on $\epsilon$. Denote by $x_{t,\epsilon} = g_{t,\epsilon}(X)$. Then one has
\begin{equation}
{\sf d}\eta_{t,\epsilon}(X) = {\sf d}x_{t,\epsilon} = u_{t,\epsilon}(x_{t,\epsilon})\,dt + \sum_{i=1}^N \xi_i(x_{t,\epsilon})\circ dW_t^i =: {\sf d}\chi_{t,\epsilon}(x_{t,\epsilon}).
\label{eq:twoparameterstochu}
\end{equation}
The vector field associated to the $\epsilon$-dependence of the two parameter subgroup is given by
\begin{equation}
\frac{\partial}{\partial \epsilon}\eta_{t,\epsilon} = \frac{\partial}{\partial \epsilon}x_{t,\epsilon} = v_{t,\epsilon}(x_{t,\epsilon}).
\label{eq:twoparameterstochv}
\end{equation}
Computing the derivative with respect to $\epsilon$ of \eqref{eq:twoparameterstochu} gives
\begin{equation}
\begin{aligned}
\frac{\partial}{\partial \epsilon}{\sf d}x_{t,\epsilon} &= \frac{\partial}{\partial \epsilon}\big({\sf d}\chi_{t,\epsilon}(x_{t,\epsilon})\big)\\
&= \left(\frac{\partial}{\partial \epsilon}u_{t,\epsilon} + v_{t,\epsilon}\cdot\frac{\partial}{\partial x_{t,\epsilon}}{\sf d}\chi_{t,\epsilon}\right)(x_{t,\epsilon}),
\end{aligned}
\end{equation}
where the independence of the data vector fields $\xi_i$ on $\epsilon$ was used. Taking the differential with respect to time of \eqref{eq:twoparameterstochv} gives
\begin{equation}
\begin{aligned}
{\sf d}\left(\frac{\partial}{\partial \epsilon} x_{t,\epsilon}\right) &= {\sf d}\big(v_{t,\epsilon}(x_{t,\epsilon})\big)\\
&=  \left( {\sf d}v_{t,\epsilon}(x_{t,\epsilon}) + {\sf d}\chi_{t,\epsilon}\cdot\frac{\partial}{\partial x_{t,\epsilon}}v_{t,\epsilon}\right)(x_{t,\epsilon}).
\end{aligned}
\end{equation}
One can then evaluate at $\epsilon=0$ and call upon equality of cross derivative-differential to obtain the result by subtracting. Since $g_{t,\epsilon}$ depends on $t$ in a $C^0$ manner, the integral representation is required. The particle relabelling symmetry permits one to stop writing the explicit dependence on space,
\begin{equation}
\delta u(t)\,dt = {\sf d}v(t) + [{\sf d}\chi_t,v(t)].
\end{equation}
This completes the proof of formula \eqref{def:delta-var} for the variation of $u(t)$.
\end{proof}
The notation in \eqref{eq:twoparameterstochu} needs careful explanation, because it comprises both a stochastic differential equation and a definition. The symbol ${\sf d}\chi_{t,\epsilon}$ is used to define a vector field, whereas ${\sf d}x_{t,\epsilon}$ denotes a stochastic differential equation. This lemma makes the presentation of the stochastic Euler-Poincar\'e theorem particularly simple.
\medskip

\begin{theorem}[Stochastic Euler-Poincar\'e theorem for the diffeomorphisms]\label{thm:SEP}
With the notation as above, the following are equivalent.
\begin{enumerate}[i)]
\item The constrained variational principle
\begin{equation}
\delta\int_{t_1}^{t_2}\ell(u,a)\,dt = 0
\end{equation}
holds on $\mathfrak{X}^s\times V^*$, using variations $\delta u$ and $\delta a$ of the form
\begin{equation}\label{eq:epconstraints}
\delta u = {\sf d}v + [{\sf d}\chi_t,v], \qquad \delta a = -\mathcal{L}_v a,
\end{equation}
where $v(t)\in \mathfrak{X}^s$ is arbitrary and vanishes at the endpoints in time for arbitrary times $t_1,t_2$.
\item The stochastic Euler-Poincar\'e equations hold on $\mathfrak{X}^s\times V^*$
\begin{equation}
{\sf d}\frac{\delta \ell}{\delta u} + \mathcal{L}_{{\sf d}\chi_t}\frac{\delta \ell}{\delta u} = \frac{\delta \ell}{\delta a}\diamond a\,dt,
\label{eq:stochep}
\end{equation}
and the advection equation
\begin{equation}
{\sf d}a + \mathcal{L}_{{\sf d}\chi_t}a = 0.
\label{eq:stochadv}
\end{equation}
\end{enumerate}
\end{theorem} 

\begin{proof}
Using integration by parts and the endpoint conditions $v(t_1)=0=v(t_2)$, the variation can be computed to be
\begin{equation}
\begin{aligned}
\delta\int_{t_1}^{t_2}\ell(u,a)\,dt 
&= 
\int_{t_1}^{t_2}\left\langle\frac{\delta\ell}{\delta u},\delta u\right\rangle + \left\langle\frac{\delta\ell}{\delta a},\delta a\right\rangle\,dt\\
&= \int_{t_1}^{t_2}\left\langle\frac{\delta\ell}{\delta u},{\sf d}v + [{\sf d}\chi_t,v]\right\rangle + \left\langle\frac{\delta\ell}{\delta a}\,dt,-\mathcal{L}_v a\right\rangle\\
&= \int_{t_1}^{t_2}\left\langle -{\sf d}\frac{\delta\ell}{\delta u} - \mathcal{L}_{{\sf d}\chi_t}\frac{\delta\ell}{\delta u} + \frac{\delta\ell}{\delta a}\diamond a\,dt,v\right\rangle\\
&= 0\,.
\end{aligned}
\end{equation}
Since the vector field $v$ is arbitrary, one obtains the stochastic Euler-Poincar\'e equations. Finally, the advection equation \eqref{eq:stochadv} follows by applying the KIW formula to $a(t)=\eta_{t*}a_0$.
\end{proof}

\begin{remark}
The stochastic Euler-Poincar\'e theorem is equivalent to the version presented in \citet{holm2015variational}, which uses stochastic Clebsch constraints. In \cite{holm2015variational} one can also find an investigation the It\^o formulation of the stochastic Euler-Poincar\'e equation. The general form of the Euler-Poincar\'e theorem uses ${\rm ad}$ in place of the commutator in \eqref{eq:epconstraints} and the derivative of the representation of the group action in place of the Lie derivative. The resulting Euler-Poincar\'e equations are then formulated using ${\rm ad}^*$ for the group that is involved. The general form can be found in \cite{holm1998euler} and will be used in section \eqref{sec:circle}.
\end{remark}

\paragraph{Stochastic Lie-Poisson formulation.}
The stochastic Euler-Poincar\'e equations have an equivalent stochastic Lie-Poisson formulation. To obtain the Lie-Poisson formulation, one must Legendre transform the reduced Lagrangian. The Legendre transformation in the presence of stochasticity becomes itself stochastic in the following way
\begin{equation}
m := \frac{\delta\ell}{\delta u}, \qquad \hslash(m,a)\,dt + \sum_{i=1}^N\langle m,\xi_i\rangle \circ dW_t^i = \langle m,{\sf d}\chi_t\rangle - \ell(u,a)\,dt.
\label{eq:reducedstochlegendre}
\end{equation}
The stochasticity enters the Legendre transformation because the momentum map $m$ is coupled to the stochastic vector field ${\sf d}\chi_t$. The left hand side of the transformation determines the Hamiltonian, which is a semimartingale. The underlying semidirect product group structure has not changed, it is still the $H^s$ diffeomorphisms with a vector space, but the Hamiltonian has become a semimartingale. This implies that in the stochastic case the energy is not conserved, because Hamiltonian depends on time explicitly. Note that \eqref{eq:reducedstochlegendre} emphasises that the Lagrangian does not feature stochasticity in this framework. Instead, the Lagrangian represents the physics in the problem, which does not change. The stochasticity is supposed to account for the difference between observed data and deterministic modelling. The stochastic Lie-Poisson equations are given by
\begin{equation}
{\sf d}(m,a) = -{\rm ad}^*_{(\frac{\delta\hslash}{\delta m},\frac{\delta\hslash}{\delta a})}(m,a)\,dt - \sum_{i=1}^N{\rm ad}^*_{(\xi_i,0)}(m,a)\circ dW_t^i,
\label{eq:stochliepoisson}
\end{equation}
where ${\rm ad}^*$ is given in \eqref{eq:ad*}. Since both the drift and the diffusion part use the same operator (the ${\rm ad}^*$ operator) in \eqref{eq:stochliepoisson}, the stochastic Lie-Poisson equations preserve the same family of Casimirs (or integral conserved quantities) as the deterministic Lie-Poisson equations. The stochastic Lie-Poisson equation \eqref{eq:stochliepoisson} is given in general form, which is convenient for the derivation of families of stochastic wave equations that we will perform in section \ref{sec:circle}. The stochastic Euler-Poincar\'e theorem has a stochastic Kelvin-Noether circulation theorem as a corollary.

\paragraph{Stochastic Kelvin-Noether theorem.} 
Let $\mathfrak{C}^s$ be the space of loops $\gamma:S^1\to\mathfrak{D}^s$, which is acted upon from the left by $\mathfrak{D}^s$. Given an element $m\in\mathfrak{X}^{s*}$, one obtains a covector-valued density whose density is constant in space by dividing the momentum $m$ by the density $\rho$. By considering only the value part of the momentum, the Kelvin-Noether theorem is as follows. The circulation map $\mathcal{K}:\mathfrak{C}^s\times V^*\to\mathfrak{X}^{s**}$ is defined by 
\begin{equation}
\langle \mathcal{K}(\gamma,a),m\rangle = \oint_\gamma\frac{m}{\rho}\,.
\end{equation}
Given a Lagrangian $\ell:\mathfrak{X}^s\times V^*\to \mathbb{R}$,  the \emph{Kelvin-Noether quantity} is defined by
\begin{equation}
I(\gamma,u,a) := \oint_\gamma\frac{1}{\rho}\frac{\delta\ell}{\delta u}\,.
\end{equation}
One can now formulate the following stochastic Kelvin-Noether circulation theorem. 
\medskip

\begin{theorem}[Stochastic Kelvin-Noether theorem]\label{thm:KelThm}
Let $u_t=u(t)$ satisfy the stochastic Euler-Poincar\'e equation \eqref{eq:stochep} and $a_t=a(t)$ the stochastic advection equation \eqref{eq:stochadv}. Let $\eta_t$ be the flow associated to the vector field ${\sf d}\chi_t$. That is, ${\sf d}\chi_t = {\sf d}\eta_t\circ \eta_t^{-1} = u_t\,dt + \sum_{i=1}^N \xi_i\circ dW_t^i$. Let $\gamma_0\in \mathfrak{C}^s$ be a loop. Denote by $\gamma_t = \eta_t\circ \gamma_0$ and define the Kelvin-Noether quantity $I(t):= I(\gamma_t,u_t,a_t)$. Then
\begin{equation}
{\sf d}I(t) = \oint_{\gamma_t}\frac{1}{\rho}\frac{\delta\ell}{\delta a}\diamond a\,dt\,.
\label{eqn:KelThm}
\end{equation}
\end{theorem}
\begin{proof}
The statement of the stochastic Kelvin-Noether circulation theorem involves a loop that is moving with the stochastic flow. One can transform to stationary coordinates by pulling back the flow to the initial condition. This pull-back yields
\begin{equation}
I(t) = \oint_{\gamma_t}\frac{1}{\rho}\frac{\delta\ell}{\delta u} = \oint_{\gamma_0}\eta_t^*\left(\frac{1}{\rho}\frac{\delta\ell}{\delta u}\right) = \oint_{\gamma_0}\frac{1}{\rho_0}\eta_t^*\left(\frac{\delta\ell}{\delta u}\right).
\end{equation}
An application of the Kunita-It\^o-Wentzell formula \eqref{eq:kiwformula} leads to 
\begin{equation}
{\sf d}I(t) = \oint_{\gamma_0}\frac{1}{\rho_0}\eta_t^*\left({\sf d}\frac{\delta\ell}{\delta u} + \mathcal{L}_{{\sf d}\chi_t}\frac{\delta \ell}{\delta u}\right) = \oint_{\gamma_0}\frac{1}{\rho_0}\eta_t^*\left(\frac{\delta\ell}{\delta a}\diamond a\right)\,dt,
\end{equation}
since $u$ satisfies the stochastic Euler-Poincar\'e theorem. Transforming back to the moving coordinates by pushing forward with $\eta_t$ yields the final result.
\end{proof}

Thus, Theorem \ref{thm:KelThm} explains how particle relabelling symmetry gives rise to the Kelvin-Noether circulation theorem via Noether's theorem. When the only advected quantity present is the mass density, the loop integral of the diamond terms vanishes. This means that circulation is conserved according to Noether's theorem for an incompressible fluid, or for a barotropically compressible fluid. The presence of other advected quantities breaks the symmetry further and introduces the  \emph{diamond terms} which generate circulation, as one can see in the Kelvin-Noether circulation theorem in equation \eqref{eqn:KelThm}. Consequently, the symmetry breaking due to additional order parameters can provide additional mechanisms for the generation of Kelvin-Noether circulation in ideal fluid dynamics. 

\paragraph{Stochastic Euler-Boussinesq equations.}
Let us now use the theory introduced above to derive the Euler-Boussinesq equations in the domain $\Omega$ shown in Figure \ref{fig:domain}, with a given bottom topography $h(x,y)$, a free surface $\zeta(x,y,t)$ and non-penetration boundary conditions on the lateral walls.
\begin{figure}[H]
\centering
\includegraphics[scale=.33]{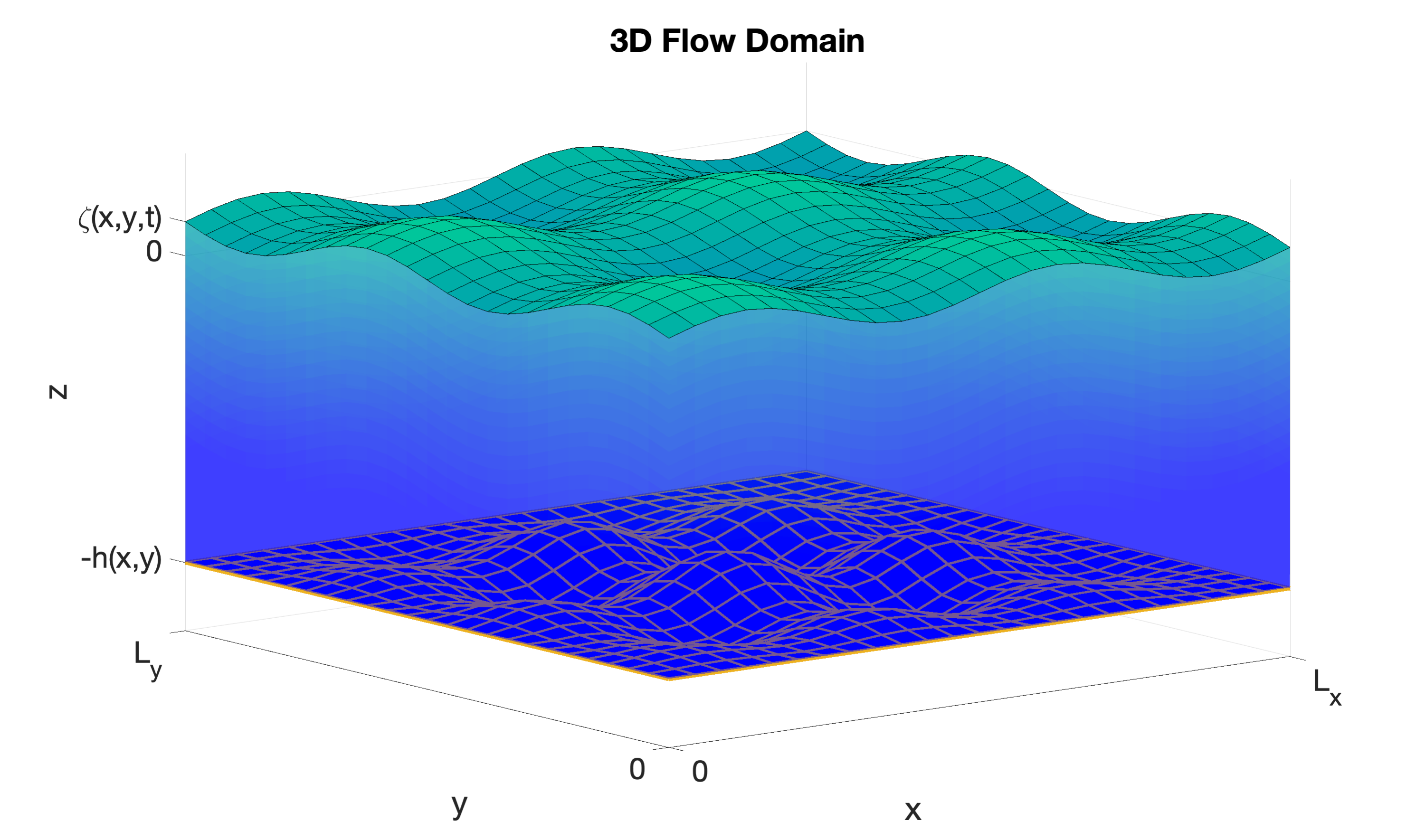}
\caption{The domain $\Omega$ for the 3D Euler-Boussinesq equations, as shown in \cite{holm2021stochastic}.}
\label{fig:domain}
\end{figure}
Let the stochastic vector field be given by 
\begin{equation}
{\sf d}\boldsymbol \chi_{3t}(x,y,z,t) := \mathbf{u}_3(x,y,z,t)\,dt + \sum_{i=1}^M \boldsymbol\xi_{3i}(x,y,z)\circ dW_t^i,
\end{equation}
where $\mathbf{u}_3=(\mathbf{u},w)$ is the three dimensional velocity field, $\mathbf{u}$ is the horizontal velocity field and $w$ is the vertical velocity. The $\boldsymbol\xi_{3i}(x,y,z)$ are data vector fields and they represent spatial velocity-velocity correlations. The $W_t^i$ are independent, identically distributed Brownian motions for each $i=1,\hdots,M$. In what follows, we will employ Einstein's convention of summing over repeated indices to shorten the notation. Before deriving the equations, let us set up the appropriate boundary conditions. When a surface in a moving fluid consists of the same particles for all time, then it is a bounding surface of the fluid. The converse is also true, every bounding surface is a material surface. Let
\begin{equation}
F(x,y,z,t) = 0
\end{equation}
be the equation for the material surface. For $F$ to be a material surface, it must be a Lagrangian invariant, i.e., the stochastic material derivative must vanish. This means that $F$ is required to satisfy
\begin{equation}
\frac{1}{|\nabla_3 F|}\big({\sf d}F + ({\sf d}\boldsymbol\chi_{3t}\cdot\nabla_3)F\big) = 0.
\end{equation}
For the free surface boundary, $F(x,y,z,t)=z-\zeta(x,y,t)$, which implies that the free surface boundary condition is given by
\begin{equation}\label{eq:freesurface}
w\,dt + \hat{\mathbf{z}}\cdot\boldsymbol \xi_{3i}\circ dW_t^i = {\sf d}\zeta + ({\sf d}\boldsymbol \chi_t\cdot\nabla)\zeta \quad \text{ at } z=\zeta(x,y,t).
\end{equation} 
The notation in \eqref{eq:freesurface} uses horizontal and vertical components of ${\sf d}\boldsymbol\chi_{3t}=({\sf d}\boldsymbol \chi_t,w\,dt + \hat{\mathbf{z}}\cdot\boldsymbol \xi_{3i}\circ dW_t^i)$. For the bottom topography, i.e., the bottom boundary condition, we set $F(x,y,z,t) = z + h(x,y)$, which yields
\begin{equation}\label{eq:bathymetry}
w\,dt + \hat{\mathbf{z}}\cdot\boldsymbol \xi_{3i}\circ dW_t^i = -({\sf d}\boldsymbol \chi_t\cdot\nabla)h \quad \text{ at } z=-h(x,y).
\end{equation}
The non-penetration boundary condition on the lateral boundaries is given by
\begin{equation}
{\sf d}\boldsymbol \chi_t\cdot\hat{\mathbf{n}} = 0, \quad \text{ on any vertical lateral boundary},
\end{equation}
where $\hat{\mathbf{n}}$ is the outward pointing unit vector normal to the lateral boundaries. This condition can be obtained from the incompressiblity condition by an application of the divergence theorem to
\begin{equation}
\nabla_3\cdot{\sf d}\boldsymbol \chi_{3t} = 0, i.e., \quad \nabla_3\cdot\mathbf{u}_3 = 0 \text{ and } \nabla_3\cdot\boldsymbol \xi_{3i}=0 \quad \forall i.
\end{equation}
For incompressible flows, the pressure is determined by the incompressibility condition, rather than by an equation of state. Since the fluid velocity field is semimartingale, the pressure also becomes a semimartingale. We use the notation
\begin{equation}
p\circ d\mathbf{S}_t = p_d\,dt + p_i\circ dW_t^i,
\end{equation}
where $\circ d\mathbf{S}_t=(dt,\circ dW_t^1,\hdots,\circ dW_t^M)$ is the notation borrowed from \cite{street2021semi} for semimartingale driven variational principles. Also, here $p_d$ corresponds to the drift part and $p_i$ for $i=1,\hdots,M$ are the pressures associated with the diffusion part. We do not include surface tension in this model. The dynamic boundary condition for the pressure is then given by
\begin{equation}
p\circ d\mathbf{S}_t = 0 \text{ at } z=\zeta(x,y,t) \quad \text{ or } \quad p\circ d\mathbf{S}_t = {\sf d}\zeta(x,y,t) \text{ at } z=0.
\end{equation}
The final boundary condition is on the buoyancy $b$, which is given by
\begin{equation}
\hat{\mathbf{n}}_3\times \nabla_3 b = 0 \text{ on } \partial\Omega
\end{equation}
This boundary condition implies that the boundary $\partial\Omega$ of the domain $\Omega$ is a level set of buoyancy. We can now derive the equations for an Euler-Boussinesq fluid inside the domain $\Omega$. The starting point is the Lagrangian 
\begin{equation}
\ell_{EB} = \int_{\Omega}\left(\frac{1}{2}|\mathbf{u}|^2 + \frac{\sigma^2}{2}w^2 + \frac{1}{{\rm Ro}}\mathbf{u}\cdot\mathbf{R} - \frac{1}{{\rm Fr}^2}(1+\mathfrak{s}b)z\right)D\,dx\,dy\,dz.
\end{equation}
Here $D$ denotes the dimensionless density of the fluid. $\mathbf{R}$ is the vector potential for the Coriolis parameter $f(x,y)$, that is, $\nabla_3\times\mathbf{R} = f(x,y)\hat{\mathbf{z}}$. Furthermore, $\sigma$ denotes the aspect ratio, ${\rm Ro}$ is the Rossby number, ${\rm Fr}$ is the Froude number and $\mathfrak{s}$ is the stratification parameter. To obtain the equations of motion we need a constrained variational principle due to the incompressibility. This means that we formulate a constrained action principle where we set the density equal to unity by a Lagrange multiplier. This Lagrange multiplier can be identified as the pressure. The dimensionless action for the Euler-Boussinesq model is given by
\begin{equation}
S_{EB} = \int_{t_1}^{t_2}\ell_{EB}\,dt - \left\langle\frac{1}{{\rm Fr}^2} p, D-1\right\rangle\circ d\mathbf{S}_t =: \int_{t_1}^{t_2} c\ell_{EB}\circ d\mathbf{S}_t,
\end{equation}
We can now compute the variational derivatives of the constrained Lagrangian $c\ell_{EB}$, which are given by 
\begin{equation}
\begin{aligned}
\frac{\delta c\ell_{EB}}{\delta \mathbf{u}} &= D\left(\mathbf{u}+\frac{1}{{\rm Ro}}\mathbf{R}\right)\\
\frac{\delta c\ell_{EB}}{\delta w} &= \sigma^2 Dw\\
\frac{\delta c\ell_{EB}}{\delta b} &= -\frac{1}{{\rm Fr}^2}Dz,\\
\frac{\delta c\ell_{EB}}{\delta D} &= \left(\frac{1}{2}|\mathbf{u}|^2 + \frac{\sigma^2}{2}w^2 + \frac{1}{{\rm Ro}}\mathbf{u}\cdot\mathbf{R} - \frac{1}{{\rm Fr}^2}(1+\mathfrak{s}b)z\right) - \frac{1}{{\rm Fr}^2} p \circ d\mathbf{S}_t,\\
\frac{\delta c\ell_{EB}}{\delta p} &= \frac{1}{{\rm Fr}^2}(D-1)\circ d\mathbf{S}_t.
\end{aligned}
\end{equation}
Inserting the variational derivatives into the stochastic Euler-Poincar\'e theorem \ref{thm:SEP} yields the stochastic Euler-Boussinesq equations
\begin{equation}
\begin{aligned}
{\sf d}\mathbf{u} + ({\sf d}\boldsymbol\chi_{3t}\cdot\nabla_3)\mathbf{u} + (\nabla\boldsymbol\xi_{3i})\cdot\mathbf{u}_3\circ dW_t^i &= -\frac{1}{{\rm Fr}^2}\nabla (p_d \,dt + p_i\circ dW_t^i) - \frac{1}{{\rm Ro}}f\hat{\mathbf{z}}\times{\sf d}\boldsymbol \chi_t - \frac{1}{{\rm Ro}}\nabla(\boldsymbol \xi_i\cdot\mathbf{R})\circ dW_t^i,\\
\sigma^2\left({\sf d}w +({\sf d}\boldsymbol\chi_{3t}\cdot\nabla_3)w + \Big(\frac{\partial}{\partial z}\boldsymbol\xi_{3i}\Big)\cdot\mathbf{u}_3\circ dW_t^i \right) &= -\frac{1}{{\rm Fr}^2}\frac{\partial}{\partial z}(p_d \,dt + p_i\circ dW_t^i) + \frac{1}{{\rm Fr}^2}(1+\mathfrak{s}b)dt,\\
{\sf d}b + ({\sf d}\boldsymbol \chi_{3t}\cdot\nabla_3)b &= 0,\\
\nabla_3\cdot{\sf d}\boldsymbol\chi_{3t} &= 0.
\end{aligned}
\end{equation}
The Kelvin circulation theorem is given by
\begin{equation}
\begin{aligned}
{\sf d}\oint_{c({\sf d}\boldsymbol \chi_{3t})}\left(\mathbf{u}_3 + \frac{1}{{\rm Ro}}(\mathbf{R},0)\right)\cdot d\mathbf{x}_3 &= -\frac{\mathfrak{s}}{{\rm Fr}^2}\oint_{c({\sf d}\boldsymbol \chi_{3t})} b\,dz\,dt\\
&= -\frac{\mathfrak{s}}{{\rm Fr}^2}\int\!\!\int_{\partial S = c({\sf d}\boldsymbol\chi_{3t})}\hat{\mathbf{z}}\times\nabla_3 b\cdot d\mathbf{S}_3\,dt.
\end{aligned}
\end{equation}
The Kelvin circulation theorem shows that for the Euler-Boussinesq model circulation is being generated whenever the buoyancy gradient does not align with the vertical unit vector. Much more can be said about the structure of these equations, such as a Lagrangian invariant known as the potential vorticity, an infinite family of conserved integral quantities known as Casimirs, for which we refer to \cite{holm2021stochastic, luesink2021stochastic}. In the next section, we will investigate the implications of the stochastic Euler-Poincar\'e equations and Lie-Poisson equations in one spatial dimension.

\section{Diffeomorphisms on the circle}\label{sec:circle}
\subsection*{Key points}
\begin{center}
    $\bullet\,$ Diffeomorphisms on $S^1$ $\qquad$
    $\bullet\,$ Virasoro-Bott group $\qquad$
    $\bullet\,$ SPDEs on $S^1$ 
\end{center}
In the previous section we discussed diffeomorphism groups over arbitrary smooth Riemannian manifolds. In dimension one, there are two smooth Riemannian manifolds. Namely, the noncompact real line, which is homeomorphic to $\mathbb{R}$ in a trivial manner, and the circle $S^1 = \mathbb{R}\setminus\mathbb{Z}$ which is locally homeomorphic to $\mathbb{R}$. Though the circle is the lowest dimensional compact manifold one can study, the diffeomorphism group over the circle has a rich mathematical nature that leads to several remarkable partial differential equations. A key difference with higher dimensional manifolds is that in dimension one there is no vorticity, since every 1-form is exact, so therefore also no Kelvin circulation theorem. This makes the one-dimensional case special.

The geodesics for a given metric of the diffeomorphism group over the circle gives rise to families of well-known wave equations. When the kinetic energy is given by the $L^2$ metric, the geodesic equation is Burgers equation, which is known to form shocks in finite time. When the kinetic energy is given by the $H^1$ metric, the geodesic equation is the Camassa-Holm equation, which is completely integrable. Changing the metric can therefore regularise equations and this approach was used with great success in the context of the $\alpha$-model in turbulence, see \cite{foias2001navier}.  

Changing the metric is a geometric resolution to a problem that also has a topological approach. The topological approach improves the behaviour of the coadjoint orbits and also leads to a completely integrable partial differential equation for the $L^2$ metric, the Korteweg-De Vries equation. The Virasoro-Bott group $Vir$ is obtained by a unique nontrivial central extension of the circle diffeomorphisms, using the Bott-Thurston cocycle. At the Lie algebra level, one has a corresponding unique central extension by the Gel'fand-Fuchs cocycle and one obtains the Virasoro algebra $\mathfrak{vir}$. In this section we discuss the adjoint and coadjoint representation theory of the Virasoro-Bott group and the Virasoro algebra. Using these representations and those of the unextended case, we derive the stochastic Burgers equation, the stochastic Korteweg-De Vries equation and the stochastic Camassa-Holm equation without dispersion and finally the stochastic Camassa-Holm equation with dispersion. All of these equations play a role in fluid dynamics. In passing, we also derive the stochastic Hunter-Saxton equation, which plays a role in the study of nematic liquid crystals.

We start by extending the Lie algebra of vector fields on the circle rather than starting by extending the group of orientation preserving diffeomorphisms of the circle. Starting with the algebra is the convention in infinite dimensions, because every infinite dimensional topological group has a corresponding Lie algebra, but the converse is not true. Fix a coordinate $\theta$ on the circle, then any vector field in $\mathfrak{X}(S^1)$ can be written as $f(\theta)\partial_\theta$ for any function $f\in C^\infty(S^1)$, where $\partial_\theta$ is shorthand notation for $\frac{\partial}{\partial \theta}$. The space of vector fields $\mathfrak{X}(S^1)$ is a Lie algebra whose bracket is given by
\begin{equation}
[f(\theta)\partial_\theta,g(\theta)\partial_\theta] = \Big(f(\theta)g'(\theta) - g(\theta)f'(\theta)\Big)\partial_\theta.
\end{equation}
The following map $\omega:\mathfrak{X}(S^1)\times\mathfrak{X}(S^1)\to\mathbb{R}$, given by
\begin{equation}
\omega\Big(f(\theta)\partial_\theta,g(\theta)\partial_\theta\Big) = \int_{S^1}f'(\theta)g''(\theta)\,d\theta,
\end{equation}
is the Gel'fand-Fuchs 2-cocycle which nontrivially extends the algebra $\mathfrak{X}(S^1)$ to the Virasoro algebra $\mathfrak{vir}$. The Gel'fand-Fuchs map is a 2-cocycle, which means that it satisfies the identity
\begin{equation}
\omega([f\partial_\theta,g\partial_\theta],h\partial_\theta) + \omega([h\partial_\theta,f\partial_\theta],g\partial_\theta) + \omega([g\partial_\theta,h\partial_\theta],f\partial_\theta) = 0.
\end{equation}
To define an extension of a Lie algebra, which is a Lie algebra itself, the resulting extended bracket is of course required to satisfy the Jacobi identity. The 2-cocycle identity above is necessary for the extended bracket to satisfy the Jacobi identity. We can then form the extended Lie bracket with $a,b\in\mathbb{R}$
\begin{equation}\label{eq:virbracket}
[(f\partial_\theta,a),(g\partial_\theta,b)]_{\mathfrak{vir}}=([f\partial_\theta,g\partial_\theta],\omega(f\partial_\theta,g\partial_\theta)),
\end{equation}
which defines the Virasoro algebra as $(\mathfrak{{vir}},[\,\cdot\,,\,\cdot\,]_{\mathfrak{vir}})=(\mathfrak{X}(S^1)\oplus\mathbb{R},([\,\cdot\,,\,\cdot\,],\omega(\,\cdot\,,\,\cdot\,))$. It can be shown that the Virasoro algebra is the unique (up to isomorphism) central extension of the algebra of vector fields on the circle. The proof of this statement can be found in \cite{khesin2008geometry}, which uses the fact that the second cohomology group $H^2(\mathfrak{X}(S^1),\mathbb{R})$ is one-dimensional and generated by the Gel'fand-Fuchs cocycle. It can also be shown that the Gel'fand-Fuchs cocycle is nontrivial, which means that it is not a coboundary. This means that there does not exist a map (called a coboundary) $\alpha:\mathfrak{X}(S^1)\to\mathbb{R}$ such that $\omega(X,Y) = \alpha([X,Y])$. In general it is not guaranteed that there corresponds a group to centrally extended Lie algebra, but for the Virasoro algebra this is the case.

The Virasoro algebra corresponds to the group known as the Virasoro-Bott group, which is the central extension of the orientation preserving diffeomorphisms on the circle $\mathfrak{D}^+(S^1)$ by the Bott 2-cocycle $B:\mathfrak{D}^+(S^1)\times\mathfrak{D}^+(S^1)\to \mathbb{R}$. Let $\varphi,\psi\in \mathfrak{D}^+(S^1)$ and let $\theta$ be the coordinate on $S^1$. By $\varphi'$ we mean a derivative of $\varphi$ with respect to $\theta$. The Bott-Thurston 2-cocycle is then given by 
\begin{equation}
B(\varphi,\psi) = \frac{1}{2}\int_{S^1}\log(\varphi'\circ\psi) d\log\psi',
\end{equation}
here given with constant $\frac{1}{2}$. The usual normalisation constant is $-\frac{1}{48\pi}$, which involves $\frac{1}{2\pi}$ to normalise the circumference of the circle. The remaining factor of $-\frac{1}{24}$ is deeply rooted in the theory of modular forms, to which the Bott-Thurston cocycle has a remarkable link. The Bott-Thurston cocycle is required to satisfy the group 2-cocycle identity
\begin{equation}
B(\varphi\circ\zeta,\psi)+B(\varphi,\zeta) = B(\varphi,\zeta\circ\psi)+B(\zeta,\psi).
\end{equation}
The group 2-cocycle identity is necessary to guarantee that the group operation on the Virasoro-Bott group is associative. The Bott-Thurston 2-cocycle is a continuous 2-cocycle on $\mathfrak{D}^+(S^1)$ and the central extension is the Virasoro-Bott group $(Vir,\bullet) = (\mathfrak{D}^+(S^1)\oplus\mathbb{R},(\circ,+))$, where the product is defined by
\begin{equation}
(\varphi,\alpha_1)\bullet(\psi,\alpha_2) = (\varphi\circ\psi,\alpha_1+\alpha_2+B(\varphi,\psi)).
\end{equation}
Note that the group action on the $\mathbb{R}$-part is addition, which is Abelian. Hence in the extended Lie bracket \eqref{eq:virbracket}, the contributions of $\alpha_1$ and $\alpha_2$ disappear. The Lie algebra corresponding to $Vir$ is the Virasoro algebra $\mathfrak{vir}$. This can be shown by performing the same procedure as in section \ref{sec:adjcoadj}, i.e., taking two one-parameter subgroups of the Virasoro-Bott group and differentiating them with respect to their parameters.

We now want to move on towards the adjoint and coadjoint representations. The smooth dual space $\mathfrak{X}(S^1)^*$ is identified with the space of covector-valued densities $\Omega^1(S^1,T^*S^1) = \{u(\theta)d\theta\otimes d\theta\}$ with the pairing given by
\begin{equation}
\langle u(\theta)d\theta\otimes d\theta , f(\theta)\partial_\theta\rangle = \int_{S^1} \big(f(\theta)\partial_\theta\contract u(\theta)d\theta\big)d\theta = \int_{S^1}f(\theta)u(\theta)\,d\theta,
\end{equation}
for any vector field $f(\theta)\partial_\theta\in\mathfrak{X}(S^1)$.  With the pairing between the Virasoro algebra and its dual defined, we can compute the adjoint and coadjoint representations following the same route as in section \ref{sec:adjcoadj}. For the unextended circle diffeomorphisms the following table lists the all the adjoint and coadjoint operators
\begin{table}[H]
\centering
\begin{tabular}{l|ll}
${\rm AD}:\mathfrak{D}(S^1)\times\mathfrak{D}(S^1)\to\mathfrak{D}(S^1)$ & ${\rm AD}_{\varphi}:$ & $ \widetilde{\varphi} \mapsto \varphi\circ\widetilde{\varphi}\circ\varphi^{-1}$ \\
${\rm Ad}:\mathfrak{D}(S^1)\times \mathfrak{X}(S^1)\to\mathfrak{X}(S^1)$ & ${\rm Ad}_{\varphi}:$ & $ f(\theta)\partial_\theta\mapsto f(\varphi)\partial_\varphi$ \\
${\rm Ad}^*:\mathfrak{D}(S^1)\times \mathfrak{X}(S^1)^*\to\mathfrak{X}(S^1)^*$ & ${\rm Ad}^*_{\varphi^{-1}}:$ & $ u(\theta)d\theta\otimes d\theta\mapsto u(\varphi)d\varphi\otimes d\varphi$ \\
${\rm ad}:\mathfrak{X}(S^1)\times\mathfrak{X}(S^1)\to\mathfrak{X}(S^1)$ & ${\rm ad}_{f\partial_\theta}:$ & $g\partial_\theta\mapsto (fg'-gf')\partial_\theta$\\
${\rm ad}^*:\mathfrak{X}(S^1)\times\mathfrak{X}(S^1)^*\to\mathfrak{X}(S^1)^*$ & ${\rm ad}^*_{f\partial_\theta}:$ & $u\,d\theta\otimes d\theta \mapsto -((uf)'+uf')d\theta\otimes d\theta$
\end{tabular}
\caption{The adjoint and coadjoint actions of the diffeomorphism group over the circle and the vector fields over the circle.}
\label{tab:diffcircle}
\end{table}
For the Virasoro-Bott group and the Virasoro algebra, the adjoint and coadjoint operators are given by
\begin{table}[H]
\centering
\begin{tabular}{l|ll}
${\rm AD}:Vir\times Vir\to Vir$ & ${\rm AD}_{(\varphi,\alpha)}:$ & $(\widetilde{\varphi},\widetilde{\alpha})\mapsto \big(\varphi\circ\widetilde{\varphi}\circ\varphi^{-1},\, \frac{1}{2} \int_{S^1}\log\big((\varphi\circ\widetilde{\varphi})'\circ\varphi^{-1}\big)d(\log\varphi^{-1})' + \widetilde{\alpha}\big)$\\
${\rm Ad}:Vir \times \mathfrak{vir}\to\mathfrak{vir}$ & ${\rm Ad}_{(\varphi,\alpha)}:$ & $(f(\theta)\partial_\theta,b)\mapsto\big(f(\varphi)\partial_{\varphi},\, b+\frac{1}{2}\int_{S^1}(f\circ\varphi^{-1})d\log(\varphi^{-1})'\big)$\\
${\rm Ad}^*:Vir\times \mathfrak{vir}^*\to\mathfrak{vir}^*$ & ${\rm Ad}^*_{(\varphi,\alpha)^{-1}}:$ & $\big(u(\theta)d\theta\otimes d\theta,a\big)\mapsto \Big(\big(u(\varphi)(\varphi')^2 + a(S\varphi)(\theta)\big)d\theta\otimes d\theta,\, a\Big)$\\
${\rm ad}:\mathfrak{vir}\times\mathfrak{vir}\to\mathfrak{vir}$ &  ${\rm ad}_{(f\partial_\theta,b)}:$ & $(g\partial_\theta,c)\mapsto \Big((fg'-gf')\partial_\theta,\, \int_{S^1}f'g''d\theta\Big)$\\
${\rm ad}^*:\mathfrak{vir}\times\mathfrak{vir}^*\to\mathfrak{vir}^*$ & ${\rm ad}^*_{(f\partial_\theta,b)}:$ & $\big(u\, d\theta\otimes d\theta,a\big)\mapsto \big(-((uf)'+uf'+af''')d\theta\otimes d\theta,\, 0\big)$
\end{tabular}
\caption{The adjoint and coadjoint actions and operators of the Virasoro-Bott group and the Virasoro algebra.}
\label{tab:virasoro}
\end{table}
Contrary to the semidirect product case in Table \ref{tab:semidirect}, the central extension variable influences the momentum variable, but plays no significant role in the coadjoint actions ${\rm Ad}^*$ and ${\rm ad}^*$. This automatically implies that the variable $a$ in the pair $(u\,d\theta\otimes d\theta, a)$ has no dynamics of itself. In the coadjoint action ${\rm Ad}^*$, it should be noted that $u(\varphi)(\varphi')^2(d\theta\otimes d\theta) = u(\varphi)(d\varphi \otimes d\varphi)$ and the operator $(S\varphi)(\theta)$ is the Schwarzian derivative of the diffeomorphism $\varphi$, which is given by
\begin{equation}
(S\varphi)(\theta) = \frac{\varphi'''(\theta)}{\varphi'(\theta)}-\frac{3}{2}\left(\frac{\varphi''(\theta)}{\varphi'(\theta)}\right)^2.
\end{equation}
The Schwarzian derivative is an operator that appears in various areas of mathematics as shown by \cite{ovsienko2009schwarzian}. The Schwarzian derivative is special because it is invariant under M\"obius transformations. This means that for $a,b,c,d\in\mathbb{R}$, $\varphi,\psi\in \mathfrak{D}(S^1)$ and $\theta\in S^1$, we have
\begin{equation}
(S\varphi)(\theta) = (S\psi)(\theta)\quad \text{ implies }\quad \psi(\theta) = \frac{a\varphi(\theta) + b}{c\varphi(\theta) + d}.
\end{equation} 
The invariance under $SL(2,\mathbb{R})$ links the Schwarzian derivative with the theory of modular forms. We now have all the tools that are necessary for the derivation of the stochastic Burger's equation, the stochastic Camassa-Holm equation, the Hunter-Saxton equation, which are the stochastic partial differential equations that arise as geodesic equations on the diffeomorphism group of the circle perturbed by stochastic advection by Lie transport (SALT). We also are ready for deriving the stochastic Korteweg-De Vries equation, the stochastic Camassa-Holm equation with dispersion and the stochastic Hunter-Saxton equation with dispersion, which arise as the geodesic equations on the Virasoro-Bott group with SALT.

\paragraph{Stochastic Hopf equation.}
One obtains the stochastic Hopf equation by taking the semimartingale Hamiltonian \eqref{eq:reducedstochlegendre} with drift given by the $L^2$ metric on the group of diffeomorphisms over the circle. Substituting the expression for ${\rm ad}^*$ associated with the diffeomorphisms over the circle given in Table \ref{tab:diffcircle} into \eqref{eq:stochliepoisson} yields 
\begin{equation}\label{eq:hopf}
\begin{aligned}
{\sf d} m + (m\,{\sf d}\chi_t)_x + m({\sf d}\chi_t)_x &= 0,\\
m&= u.
\end{aligned}
\end{equation}

\paragraph{Stochastic Korteweg--De Vries equation}
The stochastic Korteweg--De Vries equation is obtained by taking the same semimartingale Hamiltonian as for the Hopf equation \eqref{eq:hopf}, but now for the Virasoro-Bott group. Substituting the expression for ${\rm ad}^*$ associated with the Virasoro-Bott group given in Table \ref{tab:virasoro} into \eqref{eq:stochliepoisson} yields
\begin{equation}
\begin{aligned}
{\sf d} m + (m\,{\sf d}\chi_t)_x + m({\sf d}\chi_t)_x +\varepsilon({\sf d}\chi_t)_{xxx} &= 0,\\
m&= u.
\end{aligned}
\end{equation}

\paragraph{Stochastic Camassa--Holm equation}
The stochastic Camassa--Holm equation is derived by returning to the diffeomorphism group over the circle and taking the semimartingale Hamiltonian as in \eqref{eq:reducedstochlegendre}, but now with the drift given by the $H^1$ metric. Substituting this into \eqref{eq:stochliepoisson} yields the stochastic Camassa--Holm equation
\begin{equation}
\begin{aligned}
{\sf d} m + (m\,{\sf d}\chi_t)_x + m({\sf d}\chi_t)_x &= 0,\\
m &= u - \alpha^2 u_{xx}.
\end{aligned}
\end{equation}

\paragraph{Stochastic dispersive Camassa--Holm equation}
If instead of the diffeomorphisms over the circle, one switches the group to the Virasoro-Bott group, but keeps the same semimartingale Hamiltonian, one obtains the stochastic Camassa--Holm equation with dispersion
\begin{equation}
\begin{aligned}
{\sf d} m + (m\,{\sf d}\chi_t)_x + m({\sf d}\chi_t)_x + \varepsilon({\sf d}\chi_t)_{xxx} &= 0,\\
m &= u - \alpha^2 u_{xx}.
\end{aligned}
\end{equation}

\paragraph{Stochastic Hunter--Saxton equation}
The stochastic Hunter--Saxton equation is obtained by taking the semimartingale Hamiltonian with drift given by the homogeneous $\dot{H}^1$ norm, that is, the $H^1$ norm without the $L^2$ part. On the group of diffeomorphisms over the circle, one finds
\begin{equation}
\begin{aligned}
{\sf d} m + (m\,{\sf d}\chi_t)_x + m({\sf d}\chi_t)_x &= 0,\\
m &= - \alpha^2 u_{xx}.
\end{aligned}
\end{equation}

\paragraph{Stochastic dispersive Hunter--Saxton equation}
Taking the same semimartingale Hamiltonian as for the stochastic Hunter--Saxton equation, but changing the group the Virasoro--Bott group, one finds the stochastic Hunter--Saxton equation with dispersion
\begin{equation}
\begin{aligned}
{\sf d} m + (m\,{\sf d}\chi_t)_x + m({\sf d}\chi_t)_x + \varepsilon({\sf d}\chi_t)_{xxx} &= 0,\\
m &= - \alpha^2 u_{xx}.
\end{aligned}
\end{equation}

\section{Conclusion}
We are dealing with Lie group-invariant variational principles, which can be either deterministic or stochastic. Lie-symmetry reduction of Hamilton's principle produces Euler-Poincar\'e equations on the Lagrangian side and Lie-Poisson equations on the Hamiltonian side. The stochastic Lie chain rule allows us to promote these deterministic equations to create stochastic geometric mechanics. We have applied this framework to the Euler-Boussinesq system and to shallow water wave theory in 1D.

\section*{Acknowledgments}
We are grateful to many friends and colleagues who have helped us in understanding the recent development of stochastic geometric mechanics, particularly C. J. Cotter, D. Crisan, A. Bethencourt de Le\'on, S.R. Ephrati, A.D. Franken, F. Gay-Balmaz, B.J. Geurts, R. Hu, W. Pan, J.-M. Leahy, J. P. Ortega, T. Ratiu, O. D. Street, T. Tyranowski, and S. Takao. 
The work of DH presented here was partially supported by European Research Council (ERC) Synergy grant entitled ``Stochastic Transport in Upper Ocean Dynamics'', STUOD - DLV-856408.
\bibliographystyle{plainnat}
\bibliography{biblio}

\end{document}